\documentclass[11pt]{article}
\usepackage{fullpage}

\newcommand{\ignore}[1]{}

\usepackage{amsmath,amssymb,amsfonts,amsthm}
\usepackage{textcomp}
\usepackage{beton,ccfonts,concmath,euler}

\usepackage[linesnumbered,vlined,ruled]{algorithm2e}

\newtheorem{theorem}{Theorem}[section]
\newtheorem{lemma}[theorem]{Lemma}
\newtheorem{corollary}[theorem]{Corollary}

\newtheorem{challenge}{Challenge}

\theoremstyle{remark}
\newtheorem*{rem}{Remark}

\usepackage[format=plain, labelfont=bf, textfont=sl]{caption}

\newcommand{\romanitems}{\renewcommand{\labelenumi}{{\rm (}\roman{enumi}\/{\rm )}}}

\usepackage{xcolor}
\usepackage[colorlinks]{hyperref}
\hypersetup{
   linkcolor={blue!50!black},
   citecolor={blue!50!black},
   urlcolor={red!40!black}
}

\usepackage{autonum}

\newcommand{\ROABP}{\rm ROABP}

\newcommand{\F}{\mathbb{F}}
\newcommand{\K}{\mathbb{K}}
\newcommand{\N}{\mathbb{N}}
\newcommand{\Q}{\mathbb{Q}}

\newcommand{\calC}{\mathcal{C}}
\newcommand{\calD}{{\mathcal D}}
\newcommand{\calE}{{\mathcal E}}
\newcommand{\calH}{\mathcal{H}}
\newcommand{\calP}{{\mathcal P}}
\newcommand{\calQ}{{\mathcal Q}}

\renewcommand{\a}{{\boldsymbol a}}
\renewcommand{\b}{{\boldsymbol b}}
\renewcommand{\c}{{\boldsymbol c}}
\newcommand{\cb}{{\boldsymbol \omega}}
\newcommand{\e}{{\boldsymbol e}}

\newcommand{\w}{{\boldsymbol w}}

\newcommand{\z}{{\boldsymbol z}}

\newcommand{\bal}{{\boldsymbol \alpha}}
\newcommand{\bbet}{{\boldsymbol \beta}}
\newcommand{\bgam}{{\boldsymbol \gamma}}
\newcommand{\var}{\textup{var}}

\newcommand{\size}{{\rm size}}
\newcommand{\set}[2]{\left\{\,#1 \mid #2\,\right\}}
\newcommand{\divides}{|}
\newcommand{\PIT}{{\rm PIT}}
\newcommand{\Div}[2]{{\rm Div}(#1/#2)}
\newcommand{\Fac}[2]{{\rm Factor}(#1|_{#2})}
\newcommand{\Irr}{{\rm Irred\mbox{-}Proj}}
\newcommand{\poly}{{\rm poly}}
\newcommand{\sparse}{\textup{sp}}

\newcommand{\Hom}[2]{\textup{Hom}_{#1}[#2]}
\newcommand{\HomC}{\textup{Hom}}
\newcommand{\T}{\textup{T}}
\newcommand{\SU}{\textup{SU}}
\newcommand{\cd}[1]{\textup{Depth-}{#1}}
\newcommand{\prf}{\textup{Proj-Fac}}
\newcommand{\prfm}{\textup{Proj-Fac-mult}}

\title{Derandomizing Multivariate Polynomial Factoring\\ for Low Degree Factors
}

\author{Pranjal Dutta \thanks{College of Computing \& Data Science (CCDS), Nanyang Technological University (NTU) Singapore. Email: \texttt{pranjal.dutta@ntu.edu.sg}} \and Amit Sinhababu \thanks{Dept.~of Computer Science, Chennai Mathematical Institute, India. Email: \texttt{amitkumarsinhababu@gmail.com}} \and Thomas Thierauf \thanks{Ulm University, Germany. Email:~\texttt{thomas.thierauf@uni-ulm.de}}}
\date{\today}

\begin{document}

\maketitle

\begin{abstract}
For a polynomial $f$ from a class $\calC$ of polynomials,
we show that the problem to compute all the \emph{constant degree} irreducible factors of~$f$
reduces in polynomial time to polynomial identity tests (PIT) for class~$\calC$ and
divisibility tests of~$f$ by constant degree polynomials.
We apply the result to several classes~$\calC$ and obtain the constant degree factors in
\begin{enumerate}
    \item 
    polynomial time, for $\calC$ being polynomials that have only constant degree factors,
    \item 
    quasi-polynomial time, for $\calC$ being sparse polynomials,
    \item 
    subexponential time, for $\calC$ being polynomials that have constant-depth circuits.
\end{enumerate}
Result~2 and~3 were already shown by Kumar, Ramanathan, and Saptharishi
with a different proof, and their time complexities necessarily depend on blackbox PITs for a related bigger class~$\calC'$. 
Our complexities vary on whether the input is given as a blackbox or whitebox.

We also show that the problem to compute the \emph{sparse} factors of polynomial 
from a class~$\calC$ reduces in polynomial time to 
PIT for class~$\calC$,
divisibility tests of~$f$ by sparse polynomials, and
irreducibility preserving bivariate projections for sparse polynomials.
For~$\calC$ being sparse polynomials,
it follows that it suffices to derandomize irreducibility preserving bivariate projections for sparse polynomials in order to compute all the sparse irreducible factors efficiently.
When we consider factors of sparse polynomials that are sums of univariate polynomials, a subclass of sparse polynomials,
we obtain a polynomial time algorithm.
This was already shown by Volkovich with a different proof.

\end{abstract}


\section{Introduction}

The problem of \emph{multivariate polynomial factorization} asks to find the unique factorization of 
a given polynomial $f \in \F[x_1,\ldots,x_n]$ as a product of distinct irreducible polynomials over $\F$.
The problem reduces to \emph{univariate polynomial factorization} over the same field, for which a deterministic polynomial time algorithm is known over the field $\mathbb{Q}$. 
The complexity of multivariate factorization depends
on the representation of input and output polynomials. 
If we use \emph{dense representation} (where all the coefficients are listed including the zero coefficients), deterministic polynomial time algorithms for multivariate factoring are known \cite{Kaltofen85}. 
If we use \emph{sparse representation} (where only the nonzero coefficients are listed), only randomized polynomial time
algorithms are known~\cite{von1985factoring, kaltofen1990computing}. 
There are other standard representations like arithmetic circuits, 
and blackbox models, 
where the evaluations of the polynomial at any point, but the internal structure of the computation is hidden. 
Randomized polynomial time factorization algorithms are known in these models due to the classic results of Kaltofen \cite{kaltofen1989factorization} and Kaltofen and Trager \cite{kaltofen1990computing}.  
Randomization is naturally required for these models, as the more basic question of \emph{polynomial identity testing} (PIT) is not yet derandomized and multivariate circuit factorization algorithms use PIT. 
However, for various restricted arithmetic circuits we know deterministic (quasi-) polynomial- or subexponential-time PIT algorithms. 
Naturally, one can ask whether factorization can be derandomized for the corresponding circuit classes. In various problems such as special cases of arithmetic circuit reconstruction, we encounter multivariate factoring of structured or very restricted instances, where either the input polynomial or some or all factors of the input polynomial belong to a very restricted class of circuits. Towards derandomization of special cases of multivariate factoring, we focus on the following two goals.

\begin{challenge}[Promise factoring \cite{von1985factoring}]\label{qn2:promise-sparse-factoring}
Let $f=\prod_{i=1}^m {g_i}^{e_i}$,
where each $g_i$ is an irreducible polynomial with its sparsity at most~$s$. 
Design a deterministic polynomial-time algorithm to output~$g_i$.
\end{challenge}

To our surprise, we do not know a deterministic \emph{polynomial}-time algorithm even for the special case of Challenge~\ref{qn2:promise-sparse-factoring}, when the given blackbox computes the product of just \emph{two} irreducible sparse polynomials. 

Note that the factors of a sparse polynomial~$f$ might be nonsparse.
Therefore, instead of finding \emph{all} the irreducible factors, 
we only focus on those factors that are sparse.

\begin{challenge}\label{qn1:sparse-factoring}
Given a sparse polynomial~$f$,
design a deterministic polynomial-time algorithm that outputs all the sparse irreducible factors of~$f$.
\end{challenge}

Multivariate polynomial factoring has various applications, such as \emph{low-degree testing} \cite{arora2003improved}, constructions of pseudorandom generators for low-degree polynomials~\cite{bogdanov2005pseudorandom, dwivedi2024optimal}, algorithmic questions in computational algebraic geometry~\cite{huang2000extended} and many more. Blackbox multivariate polynomial factorization algorithms are extensively used in \emph{arithmetic circuit reconstruction}~\cite{shpilka2007interpolation, sinha2016reconstruction}, and polynomial equivalence testing~\cite{Kayal11, kayal2012affine, ramya2019linear}. 
There is a series of papers 
starting with~\cite{kabanets2003derandomizing}
about \emph{algebraic hardness vs.\ randomness} results
that  crucially rely on the closure of various algebraic complexity classes under factorization.
 For a general overview of multivariate factoring, 
 we refer to the recent survey of Bhargav, Dwivedi, 
 and Saxena~\cite{bhargav2025primer}. 

\paragraph{General framework of factoring.} The typical steps in factorization algorithms for a given polynomial  $f \in \calC$, from a class~$\calC$
are as follows.
\begin{enumerate}
    \item 
    Transform $f$ to a polynomial~$\widehat{f}$ that is monic in some newly introduced variable, using PIT; for a detailed discussion see below.
    \item 
    Project $\widehat{f}$ to a bivariate polynomial 
    such that the factorization pattern of~$\widehat{f}$ is maintained (by Hilbert Irreducibility Theorem; see Theorem~\ref{thm:HIT}).
    \item 
    Use the known algorithms to factor the projected bivariate polynomial.
    \item 
    Now often \emph{Hensel lifting} is used to lift the projected factors back to the real factors.
    In our setting, we avoid Hensel lifting and argue that we can use \emph{interpolation} instead.  This interpolation trick was first used in the randomized blackbox factoring algorithm of Kaltofen and Trager \cite{kaltofen1990computing}.
    \item 
    To make sure that the computed polynomials are indeed factors,
    there is a divisibility test at the end.
\end{enumerate}

We now discuss some of the crucial steps and their importance in factoring a restricted class of polynomials. In our restricted setting, we completely {\em bypass} the step of Hensel lifting. 
\paragraph{PIT.} Given an arithmetic circuit $C$, polynomial identity testing (PIT) asks to test if $C$ computes the zero polynomial. Randomization in factorization algorithms mostly stem from the fact that these algorithms use PIT as a subroutine. Further, Kopparty, Saraf and Shpilka \cite{kopparty2015equivalence} showed that derandomization of whitebox and blackbox multivariate circuit factoring reduces to derandomization of polynomial identity testing of arithmetic circuits in whitebox and blackbox settings respectively. However, we {\em do not} know if sparse factorization reduces to sparse PIT (the algorithms of \cite{kopparty2015equivalence} reduce to general arithmetic circuit PIT). Recently, there has been good progress on these questions by \cite{kumar2024deterministic, kumar2024towards,  bhattacharjee2025deterministic, bhattacharjee2025closure}. Earlier works of Volkovich \cite{volkovich2015deterministically, volkovich2017some} and Bhargava, Saraf and Volkovich \cite{bhargava2020deterministic} made progress on derandomization of several special cases of sparse multivariate factoring. Shpilka and Volkovich \cite{shpilka2010relation} proved whitebox/blackbox factoring of \emph{multilinear} polynomials in a class
$\mathcal{C}$ is equivalent to whitebox/blackbox derandomization of PIT of class~$\mathcal{C}$.  
For a comprehensive overview of derandomization questions related to multivariate polynomial factoring, we refer to the survey of Forbes and Shpilka~\cite{forbes2015complexity}.

\paragraph{Divisibility testing.} 
In a factorization algorithm, we may want to check if a candidate factor is truly a factor via \emph{divisibility testing}.
It asks to test if a polynomial~$g(\z)$ divides a polynomial~$f(\z)$. Forbes~\cite{forbes2015deterministic} showed that the divisibility testing question can be efficiently reduced to an instance of a PIT question of a model that relates to both~$f$ and~$g$, see Lemma~\ref{lem:divisibility}. Currently, we do not know any deterministic polynomial time algorithm even when~$g$ and~$f$ are both sparse polynomials. 
When~$f$ is a sparse polynomial and~$g$ is a linear polynomial, the problem reduces to polynomial identity testing of any-order read-once oblivious branching programs (ROABPs), 
see Section~\ref{sec:constant-degree-ROABP}.  
We do not know a deterministic polynomial-time algorithm, 
even for testing if a quadratic polynomial~$g$ divides a sparse polynomial.

\paragraph{Irreducibility projection.} 
For testing irreducibility of multivariate polynomials, one can use some of the effective versions of Hilbert's Irreducibility
Theorem~\cite{von1985factoring, kaltofen1985effective, kaltofen1995effective}. 
These results can be also seen as an effective version of a classical theorem of Bertini in algebraic geometry. A multivariate  irreducible polynomial $f(x,z_1,z_2,\ldots,z_n)$ may become reducible if we project the variables to make it univariate, but irreducibility is preserved if we project it in a way to make it a bivariate polynomial. Unfortunately, finding such a projection in deterministic polynomial time is hard. We do not know
how to find polynomial-sized irreducibility preserving projections even for sparse polynomials, though we know polynomial-sized hitting set generators for them. This remains an outstanding open question.

\paragraph{Special cases of multivariate factoring.}
In many settings, one has \emph{a priori} structural information about the factors of a given polynomial—
for instance, that certain (or even all) factors are linear, have bounded degree, or more generally, are sparse.
A classic source of such a structure is the theory of $U$-\emph{resultants} due to Macaulay, where one encounters families of polynomials that decompose into linear forms.
As another example, consider the standard arithmetization of a $3$-CNF formula:
it naturally produces a polynomial presented as a product of \emph{clause polynomials}, and each clause polynomial depends on at most three variables and is multilinear
Restricted factorizations of this type have also been studied in the context of \emph{algebraization}~\cite{impagliazzo2009axiomatic}.

There are many prior works on special cases of multivariate polynomial factoring. 
Koiran and Ressyare~\cite{koiran2018orbits} gave \emph{three} different randomized algorithms for testing, if a given
polynomial~$f$ as a blackbox can be completely factored into linear polynomials and output the factorization, if it exists. 
Medini and Shpilka~\cite{medini2021hitting} gave a deterministic polynomial-time algorithm in the blackbox setting for the case, 
where the polynomial can be factored into linearly independent linear polynomials.

Factorization of sparse multivariate polynomials is well-studied in symbolic computation. 
The classic work of Gathen and Kaltofen~\cite{von1985factoring} gives a randomized polynomial-time algorithm in the case when the input and output factors are sparse, and the number of irreducible factors is bounded. 
Volkovich~\cite{volkovich2015deterministically, volkovich2017some} made interesting progress on derandomization of some special cases of sparse multivariate factoring. 
Bhargava, Saraf, and Volkovich~\cite{bhargava2020deterministic} gave a deterministic polynomial time algorithm for factoring sparse polynomials with bounded individual degree. 
 Special cases of depth-4 polynomial identity testing are related to questions about sparse polynomial factorization \cite{gupta2014algebraic, volkovich2017some, bisht2022solving}.

The sequence of papers \cite{kumar2024deterministic, kumar2024towards, bhattacharjee2025deterministic} made significant progress on derandomization of factoring constant depth circuits. The recent breakthrough result by
 Bhattacharjee et al.~\cite{bhattacharjee2025closure} showed that all factors of constant depth circuits are computable by poly-sized constant depth circuits, and they gave a deterministic subexponential-time algorithm to compute those circuits.


\subsection{Our results}

We show  general results
that exhibits properties of a class~$\calC$ of polynomials,
such that we can compute the \emph{constant-degree} factors
or the \emph{sparse} factors of polynomials $f \in \calC$.
Thereby we  get a unified and simple way reprove some known results
and also some new results.


\paragraph{Constant-degree factors.}
In Theorem~\ref{thm:constant-degree}
we show that the irreducible constant-degree factors of a polynomial $f \in \calC$
with their multiplicities
can be computed in polynomial time relative to 
\begin{itemize}
    \item
    PIT for the
    top-degree homogeneous component of polynomials in~$\calC$,
    \item 
    divisibility tests of the partial derivatives of~$\calC$ by constant-degree polynomials.
\end{itemize}

In many cases,
class~$\calC$ is \emph{closed} under taking homogeneous components
and partial derivatives.
Then the above two oracles simplify to
\begin{itemize}
    \item PIT for~$\calC$,
    \item divisibility tests of~$\calC$ by constant-degree polynomials.
\end{itemize}

Considering the above general algorithmic framework,
the PIT comes from step~1 and divisibility from step~5.
The Hilbert Irreducibility Theorem (see Theorem~\ref{thm:HIT}) 
yields a randomized way to do step~2.
The main point here is that we derandomize this step for constant-degree factors.

For a specific class~$\calC$,
it suffices now to consider the complexity of the above two points
that we already described before.
For example for~$\calC$ being polynomials that have only constant degree factors,
this yields a polynomial time algorithm for computing the constant degree factors,
because PIT is in polynomial time and we can skip the divisibility test (Theorem~\ref{thm:constant-degree-promise}).
For~$\calC$ being \emph{sparse polynomials},
we get a quasi-polynomial-time algorithm (Corollary~\ref{cor:sparse-constant}), 
and
for~$\calC$ being polynomials computed by \emph{constant-depth circuits},
we get subexponential-time (Corollary~\ref{cor:constant-depth-factoring}).

The last two results were already shown by
Kumar, Ramanathan and Saptharishi~\cite{kumar2024deterministic}
with a different technique, using Hensel lifting.
However,
there is a subtle difference.
In our case,
we maintain the setting.
That is, when the input is given whitebox,
we use whitebox algorithms, and similarly for blackbox.
On the other hand, 
the factoring algorithm by Kumar, Ramanathan and Saptharishi~\cite{kumar2024deterministic} requires blackbox PIT algorithms,
even when the input is given in whitebox.
This is due to PIT for the \emph{resultant polynomial}
where the unknown factor is involved.

In the above two results,
this difference does not matter because we have the same complexity bounds
for whitebox and blackbox algorithms there.
An example where this difference matters are commutative
\emph{read-once oblivious arithmetic branching programs} ($\ROABP$).
There is a polynomial time whitebox PIT algorithm for ROABPs~\cite{raz2005deterministic}, while the best-known blackbox PIT algorithm for commutative ROABPs runs in
quasi-polynomial time~\cite{DBLP:journals/toc/GurjarKS17,guo2020improved}.
Moreover, 
it follows from work of Forbes~\cite{forbes2015deterministic} 
that the divisibility test for $\ROABP$s by \emph{linear} polynomials
reduces to a PIT for $\ROABP$s.
Hence,
we have again polynomial, resp.\ quasi-polynomial time
for the divisibility test in the whitebox, resp.\ blackbox setting.
It follows that the linear factors of polynomials computed by commutative $\ROABP$s
can be computed in polynomial time in the whitebox setting,
whereas it takes quasi-polynomial time in the blackbox setting
(Corollary~\ref{cor:roabp-linfactor}).


\paragraph{Sparse factors.}
In Theorem~\ref{thm:sparse-factors}
we show that the irreducible sparse factors of a polynomial $f \in \calC$
with their multiplicities
can be computed in polynomial time relative to
\begin{itemize}
    \item
    PIT for the
    top-degree homogeneous component of polynomials in~$\calC$,
    \item 
    divisibility tests of the partial derivatives of~$\calC$ by constant-degree polynomials,
    \item irreducible projection to bivariate polynomials.
\end{itemize}
Again,
if class~$\calC$ is \emph{closed} under taking homogeneous components
and partial derivatives,
the first two items boil down to a 
PIT for~$\calC$ and 
divisibility tests of~$\calC$ by constant-degree polynomials.

For~$\calC$ being \emph{sparse polynomials},
the difficulty lies in the third point:
only subexponential-time algorithms are known to derandomize
Hilbert Irreducibility Theorem for sparse polynomials \cite{bhattacharjee2025closure}.
The hardness stems from the fact that a sparse polynomial may have both sparse and non-sparse irreducible factors. Hence, preserving irreducibility for sparse polynomials will not preserve the factorization pattern, and therefore, it may be hard to get back the actual factor. 
Nevertheless,
Theorem~\ref{thm:sparse-factors} pinpoints the challenge to compute the sparse factors of sparse polynomials 
(Corollary~\ref{cor:sparse-factoring} and~\ref{cor:sparse-factoring-subexp}).

For polynomials that can be written as a sum of univariate polynomials,
a subclass of sparse polynomials,
the irreducible projection problem can be solved in polynomial time.
Therefore we can compute the sum-of-univariate factors of sparse polynomials
in polynomial time (Corollary~\ref{cor:su-factoring}).


\section{Preliminaries}

We take~$\F= \mathbb{Q}$ as the underlying field throughout the paper, 
although the results hold as well over fields with large characteristic.

Let~$\calP(n,d)$ be the set of $n$-variate polynomials of degree at most~$d$,
with variables $\z = (z_1, z_2, \ldots, z_n)$.
By~$\deg(f)$ we denote the total degree of~$f$.
For an exponent vector $\e = (e_1,e_2, \dots ,e_n)$,
we denote the monomial 
\[
\z^{\e} = z_1^{e_1}\, z_2^{e_2} \, \cdots \, z_n^{e_n}\,.
\]
Its degree is $||\e||_1 = \sum_{i=1}^n e_i$. 
The number of monomials of a polynomial~$f$ with nonzero coefficient is called
the \emph{sparsity of}~$f$ and  is denoted by~$\sparse(f)$.
The class of polynomials with sparsity~$s$ is denoted by
\begin{equation}\label{eq:sparse}
    \calC_{\sparse}(s,n,d) =   \set{p \in \calP(n,d)}{ \sparse(p) \leq s }.
\end{equation}


For $\deg(f) = d$,
we can write $f = \sum_{k=0}^d f_k$,
where $f_k = \Hom{k}{f}$ denotes the homogeneous component of~$f$ of degree~$k$.
For the highest degree component,
we also skip the index, i.e.~we define $\HomC[f] = \Hom{d}{f} = f_d$.
For a class~$\calC$ of polynomials, we define
\begin{equation}\label{eq:Hom(C)}
    \HomC[\calC] = \set{\HomC[f]}{f \in \calC}.
\end{equation}


We define class~$\partial \calC$ as all the partial 
derivatives of polynomials in~$\calC$,
\begin{equation}\label{eq:partial(C)}
    \partial \calC = \set{\frac{\partial^e f}{\partial z^e} }{f \in \calC,~z \text{ a variable of } f, \text{ and } e \geq 0}.
\end{equation}
Note that $f \in \partial \calC$, because $f = \frac{\partial^0 f}{\partial z^0} $.


Polynomial~$f(\z)$ \emph{depends on} variable~$z_i$,
if $\frac{\partial f}{\partial z_i} \not= 0$.
The variables that~$f$ depends on are denoted by~$\var(f)$,
\begin{equation}\label{eq:varf}
   \var(f(\z)) = \set{z_i}{f \text{ depends on } z_i}.
\end{equation}


A polynomial~$f$ is called \emph{irreducible},
if it cannot be factored into the product of two non-constant polynomials.


Let~$x$ and $\z = (z_1,z_2, \dots, z_n)$ be variables and~$f(x,\z)$
be a $(n+1)$-variate polynomial.
Then we can view~$f$ as a univariate polynomial $f= \sum_i a_i(\z)\, x^i$
over~$\K[x]$, where $\K = \F[\z]$.
The \emph{$x$-degree of}~$f$ is denoted by $\deg_x(f)$.
It is the highest degree of~$x$ in~$f$.
Polynomial~$f$ is called \emph{monic in}~$x$,  
if the coefficient~$a_{d_x}(\z)$ is the constant~$1$ polynomial, 
i.e.~$a_{d_x}(\z) =1$, 
where $d_x = \deg_x(f)$.


An algorithm runs in \emph{subexponential time},
if its running time on inputs of length~$n$ is bounded by 
$2^{O(n^{\varepsilon})}$,
for any $\varepsilon > 0$.


\subsection{Computational problems and complexity measures}

For classes~$\calP,\calQ$ of multivariate polynomials, we define the following
computational problems.

\begin{itemize}
    \item 
    $\PIT(\calP)$: given $p \in \calP$, decide whether $p \equiv 0$.
    \item 
    $\Fac{\calP}{\calQ}$: given $p \in \calP$, compute all its irreducible factors in~$\calQ$ with their multiplicities.
    \item 
    $\Div{\calP}{\calQ}$: given $p \in \calP$ and $q \in \calQ$, decide whether $q \divides p$.
\end{itemize}

The time complexity to solve these problems 
we denote by~$\T_{\PIT(\calP)}$, $\T_{\Fac{\calP}{\calQ}}$, and~$\T_{\Div{\calP}{\calQ}}$, respectively.

The problems are further distinguished according to the access to the given polynomial~$p$.
In the \emph{whitebox} case,
one gets the representation of~$p \in \calP$, for example a circuit or a formula,
whereas in the \emph{blackbox} case, one can only evaluate~$p$.
For PIT,
a blackbox solution therefore means to compute a \emph{hitting set} for a class~$\calP$.
This is a set $H \subseteq \F^n$ such that for every nonzero~$p \in \calP$
there exists $\a \in H$ such that $p(\a) \ne 0$.

The blackbox PIT-algorithm then simply evaluates~$p$ at all points $\a \in H$.
Hence,
the size of~$H$ is crucial for the running time~$\T_{\PIT(\calP)}$.
The \emph{trivial hitting set} for~$\calP(n,d)$ 
used in the Schwartz-Zippel PIT-Lemma
is
\begin{equation}\label{eq:hit}
    H_{n,d} = [d+1]^n
\end{equation}
of size $|H_{n,d}| = (d+1)^n$,
i.e.~exponential in the number~$n$ of variables.
We will use it for constant-variate polynomials.
Then the size is polynomial.

For the decision problems~ $\PIT(\calP)$ and~$\Div{\calP}{\calQ}$
there is an associated \emph{construction problem}.
In case of~$\PIT(\calP)$,
a construction algorithm
also yields a point~$\a \in (\F\setminus\{0\})^n$ such that $p(\a) \not= 0$,
in case when $p \not\equiv 0$.

\begin{lemma}\label{lem:PIT-construct}
Let $p \in \calP$ be a nonzero polynomial.
A point~$\a \in (\F\setminus\{0\})^n$ such that $p(\a) \not= 0$    
can be computed in time~$nd\, \T_{\PIT(\calP)}$.
\end{lemma}

\begin{proof}
In the blackbox case,
let~$H_0$ be the queries of the decision algorithm 
on input of the zero-polynomial.
Note that~$H_0$ is a hitting set for the whole class~$\calP$.
Hence,
we can find~$\a \in H_0$ in time~$\T_{\PIT(\calP)}$ such that $p(\a) \not= 0$.

Still, some coordinates of~$\a$ might be~$0$.
In this case,
we shift~$\a$:
Let~$t$ be a new variable and consider $\a+t = (a_1 +t, a_2+t, \dots, a_n+t)$.
Then $p(\a+t)$ is a nonzero polynomial in one variable of degree~$d$.
Choose a set $S \subseteq \F - \{-a_1, -a_2, \dots, -a_n\}$ of size $|S| = d+1$.
Then $\a+t$ has only non-zero coordinates, for all $t \in S$,
and there is a $t \in S$
such that $p(\a+t) \not= 0$.
For the time complexity to find the right~$t$,
we have to add~$(d+1)$ evaluations of~$p$.

In the whitebox case,
one can search for~$\a$ by assigning values successively to the variables
and do kind of a \emph{self-reduction}.
For each variable, one tries at most~$d$ values from $\{1,2, \dots, d\}$ 
for a polynomial of degree~$d$. 
If they all give~$0$, definitely $d+1$ works because it cannot be zero at~$(d+1)$ many values.
With $n$ variables, this amounts to~$nd$ calls to the $\PIT$-decision algorithm.
\end{proof}

For time complexity,
we assume that the polynomials are given in some model of computation,
such as circuits, branching programs, or formulas.
With each model, we associate a
complexity measure $\mu : \F[\z] \rightarrow \N$.
For example, let $f \in \F[\z]$, some of the commonly used measures in the literature are: 
\begin{itemize}
    \item $\mu(f) = \sparse(f)$, the number of monomials with nonzero coefficients,
    \item $\mu(f) = \size_{\Delta}(f)$, the size of the \emph{smallest} depth-$\Delta$ circuit that computes $f$,
    \item $\mu(f) = \size_{\ROABP}(f)$, 
    the \emph{width} of the smallest read-once oblivious branching program (ROABP) that computes~$f$. 
\end{itemize} 

We define classes of polynomials of bounded measure,
\begin{equation}
    \calC_{\mu}(s,n,d) \;=\; \set{f \in \calP(n,d)}{\mu(f) \le s}. 
\end{equation}
When we skip the index~$\mu$, we just refer to circuit size,
\begin{equation}
     \calC(s,n,d) = \set{p \in \calP(n,d)}{p \text{ has a circuit of size } s}.
\end{equation}
We also consider polynomials that can be computed by  circuits of size~$s$ and depth~$t$,
\begin{equation}\label{eq:bounded-depth}
    \calC_{\cd{t}} (s,n,d) = \set{p \in \calC(s,n,d)}{p \text{ has a circuit of depth } t}.
\end{equation}

We generally assume that all polynomials in this paper
can be \emph{efficiently evaluated} at any point $\a \in \F^n$
within the respective measure,
where we consider the unit-cost model for operations over~$\F$.
This holds for all the computational models usually considered in the literature.


\subsection{Transformation to a monic polynomial}

Algorithms for factoring polynomials often assume that the given polynomial is monic.
If this is not the case for the given polynomial~$f$,
we apply a transformation~$\tau$ to~$f$ that yields
a monic polynomial~$\tau(f)$ that we can factor.
From the factors of~$\tau(f)$ we can then reveal the factors of~$f$.
Although this is standard in the literature, 
we state it in the terms we introduced above.

\begin{lemma}[Transformation to monic]\label{lem:monic}
Let $\calC_{\mu} = \calC_{\mu}(s,n,d)$ be a class of polynomials,
$f(\z) \in \calC_{\mu}$, and
$f_d = \Hom{d}{f}$ be the homogeneous degree~$d$ component of~$f$. 
For a new variable~$x$, and $\bal = (\alpha_1, \dots, \alpha_n) \in (\F\setminus\{0\})^n$,
define a linear transformation~$\tau_{\bal}$  on the variables~$z_i$:
\[
\tau_{\bal}:~~z_i \mapsto \alpha_i x+z_i ,
\]
for $i = 1,2,\dots,n$.
Let $f_{\bal}(x,\z)$  be the resulting polynomial.

We can compute~$\bal$  such that 
$\frac{1}{f_d(\alpha)} \, f_{\alpha}(x,\z)$ is monic in~$x$
in time \[
nd \, \T_{\PIT(\HomC[\calC])} + \poly(snd) \,.
\]
\end{lemma}

\begin{proof}
Let~$f(\z) \in \calC$ be a polynomial of degree~$d$  
with $n$ variables $\z = (z_1,z_2, \dots, z_n)$.
To see what the transformation does,
let 
\[
f = f_0 + f_1 + \cdots + f_d,
\]
where $f_k = \Hom{k}{f}$, the homogeneous degree-$k$ component of~$f$.
Consider the degree-$d$ component, 
\[
f_d(\z) = \sum_{|\bbet|_1=d} c_{\bbet} \z^{\bbet}.
\]
Then, for~$f_{\bal}$,
we have $\deg_x(f_{\bal}) = d$ and 
the coefficient of the leading $x$-term~$x^d$ in~$f_{\bal}$ is 
$ f_d(\bal) = \sum_{|\bbet|_1=d} c_{\bbet} \bal^{\bbet}$.

Hence, the PIT algorithm for the homogeneous component~$f_d$ of~$f$
yields an~$\bal \in (\F\setminus\{0\})^n$ such that $ f_d(\bal) \not=0$,
by Lemma~\ref{lem:PIT-construct}.
Then the polynomial $\frac{1}{f_d(\bal)} \, f_{\bal}(x,\z)$ is monic in~$x$.
\end{proof}
For simplicity of notation,
assume in the following that $ f_d(\bal) = 1$, 
so that $f_{\bal}(x,\z)$ is monic in~$x$.

Since we work with the shifted polynomial, we need to ensure that the shift of variables does not affect the irreducibility of the factors; this is guaranteed by the following lemma. It is quite standard in the literature; for a nice proof, see~\cite[Lemma~B7]{kumar2024towards}.
\begin{lemma}\label{lem:shift-irreducible}
Let $f(\z) \in \F[\z]$ be an n-variate irreducible polynomial. Then, for every $\a \in \F^n$, the polynomial $f(\a x + \z)$ is also irreducible.
\end{lemma}


\subsection{Basics of factoring and interpolation and PIT}
Berlekamp~\cite{berlekamp1970factoring} 
and Lenstra, Lenstra and Lovász~\cite{lenstra1982factoring} gave efficient deterministic factorization algorithms for \emph{univariate} polynomials over finite fields and~$\Q$, respectively.
Kaltofen~\cite{kaltofen1985polynomial} showed how to reduce the factorization of \emph{bivariate}
polynomials to univariate polynomials.
In fact, the reduction works for $k$-variate polynomials,
for any constant~$k$.
In our case,
we use it for the case $k=3$.

Via standard interpolation, one can assume that the input is given as a dense representation.

\begin{lemma}[Trivariate Factorization]\label{lemma:trivariate-factoring}
Let $f(x,y,z)$ be a trivariate  polynomial of degree~$d$. 
Then there exists an algorithm that outputs all its irreducible factors and their multiplicities in time~$\poly(d)$.
\end{lemma}

The following lemma shows how to find the multiplicity of an irreducible
factor~$g$ of a polynomial~$f$. It holds when char$(\F)=0$, or large. For a concise proof, see~\cite[Lemma 4.1]{kumar2024deterministic}.

\begin{lemma}[Factor multiplicity]\label{lem:factor-multiplicity-reduction}
  Let $f(\z), g(\z) \in \F[\z]$ be non-zero polynomials and let $z \in \{z_1,z_2, \cdots, z_n\}$ be such that $\partial_z(g) \neq 0$ and~$g$ is irreducible. 
  Then the multiplicity of~$g$ in~$f$ is 
  the smallest non-negative integer~$e$ such that $g \nmid \frac{\partial^e f}{\partial z^e}$. 
\end{lemma}

Klivans and Spielman~\cite{klivans2001randomness} derandomized the isolation lemma for PIT of sparse polynomials.
Their algorithm works over fields of~$0$ or large characteristic.

\begin{theorem}[Sparse PIT and interpolation \cite{klivans2001randomness}]\label{thm:sparse-interpolation}
Let $\calC_{\sparse} = \calC_{\sparse}(s,n,d)$.
Then $\PIT(\calC_{\sparse})$ can be solved in time
\[
\T_{\PIT(\calC_{\sparse})} =  \poly(snd).
\]

Furthermore, 
given $f \in \calC_{\sparse}$,
in time~$\poly(snd)$  one can compute a set of \emph{evaluation points} $\calE \subseteq \F^n$
of size~$\poly(snd)$
such that given the evaluations of~$f$ at all points in~$\calE$,
one can solve for the coefficients of~$f$ in time~$\poly(snd)$. 
\end{theorem}

Limaye, Srinivasan, and Tavenas~\cite{limaye2022superpolynomial} designed a deterministic subexponential-time  PIT for constant-depth circuits.

\begin{theorem}[PIT for constant depth circuits~{\cite[Corollary 6]{limaye2022superpolynomial}}] \label{thm:PIT-constant-depth}
Let $\epsilon >0$ be a real number. 
Let $\calC_{\cd{t}} = \calC_{\cd{t}}(s,n,d)$ be such that
$s \le \poly(n)$ and $t = o(\log \log \log n)$.
Then $\PIT(\calC_{\cd{t}})$ can be decided in time
\[
\T_{\PIT(\calC_{\cd{t}})} = \left(n s^{O(t)} \right)^{O((sd)^{\varepsilon})}.
\]
\end{theorem}


\subsection{Divisibility testing reduces to PIT}

Given a polynomial~$f$ to factor,
our algorithms might compute a polynomial~$g$ that is a
candidate for a factor, but in fact, is not a factor.
Hence,
we have to verify whether~$g$ is a factor of~$f$,
i.e., whether $g \divides f$.
Therefore,
we are interested in the complexity of division algorithms.

When $f,g$  can be computed by circuits of size~$s$,
Strassen~\cite{strassen1973vermeidung} showed that if~$g \divides f$, 
then $h = f/g$ can be computed by a circuit of size~$\poly(sd)$, 
where $d= \deg(h)$. 
Forbes~\cite{forbes2015deterministic} observed that 
even in the case when $g \not\divides \, f$,
one can follow Strassen's argument and obtain
a small size circuit that computes a polynomial~$\widetilde{h}$
such that $g \divides f \iff f =g \widetilde{h}$.
Hence,
we have a reduction from divisibility testing to PIT.
\begin{lemma}[Divisibility reduces to PIT {\cite[Corollary 7.10]{forbes2015deterministic}}] \label{lem:divisibility}
Let $g(\z)$ and $f(\z)$ be two polynomials of degree at most~$d$. 
Let 
$S = [2d^2+1]$ and $\bal \in \F^n$ such that $g(\bal) \ne 0$. 
Then there are constants $\{c_{\beta,i}\}_{\beta \in S, 0 \le i \le d}$, 
computable in time~$\poly(d)$,
such that  for
\begin{equation}\label{eq:h-tilde}
    \widetilde{h}(\z) = \sum_{\beta \in S} f(\beta \z + \bal) \sum_{0 \le i \le d} c_{\beta,i} \, g(\beta \z + \bal)^i,
\end{equation}
we have
\[
g(\z) \divides f(\z) \iff f(\z+\bal) = g(\z+\bal)\, \widetilde{h}(\z).
\]
\end{lemma}

A consequence from Lemma~\ref{lem:divisibility} is that
divisibility testing of a polynomial computed by constant-depth circuit
by a sparse polynomial is in subexponential time.

\begin{corollary}[Constant depth by sparse division]\label{cor:constant-depth-by-sparse}
Let $\calC_{\cd{t}} = \calC_{\cd{t}}(s,n,d)$ and
$\calD_{\sparse} = \calC_{\sparse}(s,n,d)$.
For any  $\varepsilon >0$, we have that
$\Div{\calC_{\cd{t}}}{\calD_{\sparse}}$ can be decided in time
\[
\T_{\Div{\calC_{\cd{t}}}{\calD_{\sparse}}} = \left(n \, (sd)^{O(t)} \right)^{O((sd)^{\epsilon})}\,.
\]
\end{corollary}

\begin{proof}
    We apply Lemma~\ref{lem:divisibility} with  
    $g \in \calC_{\sparse}$ and $f \in \calC_{\cd{t}}$.
    Then~$\widetilde{h}$ in~(\ref{eq:h-tilde}) can be computed by a circuit of size $O(sd^2)$ and depth~$t+2$. 
    Therefore, the polynomial 
    \[
    \widetilde{f}= f(\z+\bal) -g(\z+\bal)\cdot\widetilde{h}(\z)
    \]
    can be computed by a circuit of size~$O(sd^2)$ and depth~$t+4$. 
    By Lemma~\ref{lem:divisibility}, 
    we have $\widetilde{f} =0$ iff $g \divides f$, 
    and by Theorem~\ref{thm:PIT-constant-depth}, 
    the identity can be checked  in time 
    $\left(n \, (sd)^{O(t)} \right)^{O((sd)^{\epsilon})}$.
\end{proof}

We also consider divisibility of sparse polynomials by \emph{constant-degree} polynomials.
Building on Lemma~\ref{lem:divisibility},
Forbes~\cite{forbes2015deterministic} reduced the problem to a PIT
that can be solved in quasi-polynomial time.

\begin{corollary}[Sparse by constant degree division {\cite[Corollary~7.17]{forbes2015deterministic}}]\label{cor:sparse-by-constant-degree}
Let $\calC_{\sparse} = \calC_{\sparse}(s,n,d)$ and
$\calD_{\delta} = \calP(n,\delta)$.
Then $\Div{\calC_{\sparse}}{\calD_{\delta}}$ can be decided in time
\[
\T_{\Div{\calC_{\sparse}}{\calD_{\delta}}} = (snd)^{O(\delta \log s)} \,.
\]
\end{corollary}


\subsection{Effective Hilbert's Irreducibility Theorem}\label{sec:Hilbert}

Factorization algorithms often start with
an effective version of Hilbert's Irreducibility Theorem 
due to Kaltofen and von zur Gathen.
It shows how to project a multivariate irreducible polynomial down to two variables,
such that the projected bivariate polynomial stays irreducible. 
The proof shows the existence of 
an \emph{irreducibility certifying polynomial}~$G(\b,\c)$ in~$2n$ variables corresponding to the irreducible polynomial~$g(x,\z)$. 
The nonzeroness
of~$G$ proves the irreducibility of~$g(x,\z)$ and also gives a way to find an irreducibility-preserving projection to bivariate
(see~\cite{kaltofen1985effective, kaltofen1995effective, kopparty2015equivalence}).  

\begin{theorem}\label{thm:HIT}
Let $g(x,\z)$ be an irreducible polynomial of total degree~$\delta$ with $n+1$ variables 
that is monic in~$x$. 
There exists a nonzero polynomial~$G(\b,\c)$ of degree~$2\delta^5$ in~$2n$ variables such that for $\bbet,\bgam \in \F^n$,
\begin{equation}\label{eq:Hilbert}
    G(\bbet,\bgam) \neq 0 \implies \widehat{g}(x,t) = g(x, \bbet t + \bgam) 
\text{ is irreducible,}
\end{equation}
where $g(x, \bbet t + \bgam) =  g(x, \beta_1 t + \gamma_1,  \dots, \beta_n t + \gamma_n)$. 
\end{theorem}

The certifying polynomial~$G$ immediately yields a randomized algorithm
to construct the irreducible projection~$\widehat{g}$ via PIT.
The derandomization of Hilbert's Irreducibility Theorem is a challenging open problem
in general.
Essentially it means to find a hitting set for~$G$.

We define a corresponding computational problem.
Let $\calC \subseteq \calP(n,d)$ be a class of polynomials and $g(\z) \in \calC$.
Assume we have already computed an~$\bal \in (\F\setminus\{0\})$
as in Lemma~\ref{lem:monic} that the shifted polynomial~$g_{\bal}(x,\z)$ is monic.
Now we want to find a hitting set for~$g_{\bal}$ according to~(\ref{eq:Hilbert}).

\begin{itemize}
\item 
$\Irr(\calC)$:\\ 
Given $\bal \in (\F\setminus\{0\})^n$, 
compute a set $H_{\bal} \subseteq \F^{2n}$ such that for all $g \in \calC$,
where~$g(\bal x + \z)$  is monic in~$x$, we have
\[
g \text{ irreducible }   \implies \exists (\bbet,\bgam)\in H_{\bal}~~~g(\bal x + \bbet t+\bgam) \in \F[x,t] ~\text{ is irreducible}.
\]
\end{itemize}
By $\T_{\Irr(\calC)}$
we denote the time complexity to compute~$\Irr(\calC)$.

\begin{rem}
A subtlety in the definition of $\Irr(\calC)$ is that
it asks to compute a hitting set~$H_{\bal}$ for the \emph{shifted} polynomial~$g_{\bal}$.
If~$\calC$ is closed under variable shifts,
then the hitting set is for polynomials in~$\calC$.
However,
for example~$\calC_{\sparse}$ is not closed under variable shifts.
Hence,
$\Irr(\calC_{\sparse})$ asks to compute a hitting set for
a non-sparse polynomial that comes from a shift of a sparse polynomial.
\end{rem}


\subsection{Isolation}\label{sec:isolation}

Let $M_{\delta}$ be the set of monomials in~$n$ variables 
$\z = (z_1, z_2, \dots, z_n)$
of degree bounded by~$\delta$,
\begin{equation}
M_{\delta} = \set{\z^{\e}}{||\e||_1 \leq \delta}.
\end{equation}
Note that $M_{\delta}$ is polynomially bounded, for constant~$\delta$,
\begin{equation}\label{eq:M_delta-bound}
|M_{\delta}| 
\leq \binom{n+\delta}{\delta} 
\leq (n+\delta)^{\delta} 
\leq (\delta+1)\, \delta^\delta \, n^{\delta} 
= O(n^{\delta}).
\end{equation}

There is a standard way to map the multivariate monomials in~$M_{\delta}$
in a injective way to univariate monomials of polynomial degree.
For completeness,
we describe the details.

Consider the standard \emph{Kronecker substitution} on~$M_{\delta}$.
Define
\begin{equation}
    \varphi: \; z_i \;\mapsto\; y^{(\delta+1)^{i-1}}\, .
\end{equation}
By extending~$\varphi$ linearly to monomials $\z^{\e} \in M_{\delta}$,
we get
\begin{equation}
     \varphi: \; \z^{\e} \;\mapsto\; y^{\sum_{i=1}^n e_i(\delta+1)^{i-1}}\, ,
\end{equation}
Clearly,
$\varphi$ is injective on~$M_{\delta}$.
However,
the degree of~$y$ can be exponentially large, up to~$(\delta + 1)^n$.
A way around is to take the exponents modulo some small prime number~$p$.
We have to determine~$p$ in a way to keep the mapping injective on~$M_{\delta}$.
Hence,
for any two terms $y^e, y^{e'}$ we get from~$\varphi$,
we have to ensure that $e \not\equiv e'\pmod{p}$.
Equivalently $p \not| \, (e-e')$.

We have $|e-e'| \leq (\delta + 1)^n$ and, by~(\ref{eq:M_delta-bound}), 
there are $O(n^{2\delta})$ many pairs~$e,e'$ we get from~$M_{\delta}$ via~$\varphi$.
Prime~$p$ should not divide any of these differences,
and hence, $p$ should not divide the product~$P$
of all these differences.
The product~$P$ is bounded by
\begin{equation}
P 
\leq   \left( (\delta + 1)^n \right)^{O(n^{2\delta})} 
= (\delta + 1)^{O(n^{2\delta +1})}.
\end{equation}
Hence, 
$P$ has at most 
$\log P \leq R = O(n^{2\delta + 1})$ 
many prime factors.
By the Prime Number Theorem,
there are more than~$\log P$ primes in~$[R^2]$.
Hence,
we can find an appropriate prime~$p  \leq R^2 = n^{O(\delta)} $.

\begin{lemma}\label{lem:isolation}
There is a prime~$p= n^{O(\delta)} $ such that the linear extension of
\begin{equation}\label{eq:varphi}
     \varphi_p: \; z_i \;\mapsto\; y^{w_i}\,,~~\text{where}~w_i = (\delta+1)^{i-1}\bmod{p}\,, ~\text{ for } i=1,2,\dots,n,
\end{equation}
to monomials is injective on~$M_{\delta}$.
Moreover,
we can find such a~$p$ in time~$n^{O(\delta)}$
and compute and invert~$\varphi_p$ in time~$n^{O(\delta)}$.
\end{lemma}

\begin{proof}
We already argued about the existence of prime~$p$.
For the running time,
recall that   $|M_{\delta}| = O(n^{\delta})$.
Therefore we can search for~$p$ and check whether it works on~$M_{\delta}$
in time~$n^{O(\delta)}$.
At the same time we can compute pairs of exponents~$(\e,k)$
such that $\varphi_p(\z^{\e}) = y^k$.
These pairs can be used to invert~$\varphi_p$.
\end{proof}

The mapping~$\varphi_p$ in Lemma~\ref{lem:isolation} maintains
factors of degree~$\delta$ of a polynomial in the following sense.

\begin{lemma}\label{lem:varphi-factors}
Let polynomial~$f(\z)$  factor as $f = gh$,
where~$g(\z)$ has degree~$\delta$.
Let~$\varphi_p$ be the map from Lemma~\ref{lem:isolation}.
Then we have $\varphi_p(f) = \varphi_p(g) \varphi_p(h)$,
and~$g$ can be recovered from~$\varphi_p(g)$ in time~$n^{O(\delta)}$.
\end{lemma}

Note that in Lemma~\ref{lem:varphi-factors},
we do \emph{not} claim that irreducibility is maintained:
when~$g$ is irreducible, still~$\varphi_p(g)$ might be reducible. Consider the example $n=\delta=2$. 
The weights $\{1,3\}$ make sure that each monomial in 
$M_2 = \{\, z_1,~ z_2,~ z_1^2,~ z_1z_2,~ z_2^2\,\}$ 
gets mapped to a distinct power in~$y$. 
Let 
$g(\z)= 1-z_1z_2$. 
Observe that~$g$ is irreducible, 
however 
$g(y,y^3) = (1-y^2)(1+y^2)$
is \emph{reducible}.

We combine Lemma~\ref{lem:isolation} and Theorem~\ref{thm:HIT}
to obtain a projection of a multivariate polynomial
to a $3$-variate polynomial that maintains irreducibility
of polynomials up to degree~$\delta$.

\begin{corollary}\label{cor:HIT-g}
Let $g(x,\z)$ be an irreducible polynomial of constant degree~$\delta$ with $n+1$ variables 
that is monic in~$x$. 
There exists  $\w,\w' \in \F^n$ with $w_i,w_i' = n^{\poly(\delta)}$ such that
\begin{equation}\label{eq:Psi}
    \Psi(g) = g(x,y^{w_1} t + y^{w_1'}, \dots, y^{w_n} t + y^{w_n'}) \in \F[x,y,t]
\end{equation}
is irreducible. 
Moreover,
we can compute and invert~$\Psi(g)$ in time~$n^{\poly(\delta)}$.
\end{corollary}

\begin{proof}
Let $G(\a,\b)$  be the polynomial of degree~$2\delta^5$ in~$2n$ variables
provided by Theorem~\ref{thm:HIT} for~$g$.
Let $\w,\w' \in \F^n$ with $w_i,w_i' = n^{\poly(\delta)}$ be the exponents we get from Lemma~\ref{lem:isolation} for~$G$.
That is, 
\[\widehat{G}(y) = G(y^{w_1}, \dots, y^{w_n}, y^{w'_1}, \dots, y^{w'_n}) \not= 0\;.\] 

Now, suppose that~$\Psi(g)$ is reducible.
Then it would also be reducible at a point~$y=\alpha$, 
where $\widehat{G}(\alpha) \not= 0$.
But then $\widehat{g} (x,t) = \Psi(g)(x,\alpha,t)$ would be reducible too,
and this would contradict Theorem~\ref{thm:HIT}.
We conclude that~$\Psi(g)$ is irreducible.

For the complexity,
we first determine prime~$p$ from Lemma~\ref{lem:isolation}
and then get the weights~$\w,\w'$ from above.
For a given $g(x,\z) = \sum_{k,\e} c_{k,\e} x^k \z^{\e}$,
we can compute~$\Psi(g)$ in time~$n^{\poly(\delta)}$.
For a monomial of~$g$, 
the mapping looks as follows:
\begin{equation}\label{eq:Psi-monomial}
c_{k,\e}\, x^k\, \z^{\e} ~\mapsto~
c_{k,\e}\, x^k\, \prod_{i=1}^n (y^{w_i} t + y^{w_i'})^{e_i} \,.
\end{equation}

For inversion,
we have given $h \in \F[x,y,t]$ 
which is monic in~$x$ with $x$-degree $\le \delta$.
We either have to compute~$g$ such that $\Psi(g) = h$,
or detect that~$h$ is \emph{not} in the codomain of~$\Psi$.

We set $t=0$,
i.e.\ we consider~$h(x,y,0)$.
From~(\ref{eq:Psi-monomial})
we see that monomials then must have the form
\[
c_{k,\e}\, x^k\, y^{\sum_{i=1}^n e_i w'_i}\,.
\]
From these we get the exponents~$k$ and~$\e$ 
for monomial~$x^k \, \z^{\e}$,
similar as in the proof of Lemma~\ref{lem:isolation}.
In case we thereby get a degree  $>  \delta$,
then~$h$ is not in the codomain of~$\Psi$.
Otherwise, 
we have got a candidate~$g$ of degree~$\delta$.
But we still have to check whether $\Psi(g)=h$, 
because the inversion procedure ignores the variable~$t$.
The running time for inversion is~$n^{\poly(\delta)}$.
\end{proof}

The polynomial~$g$ of degree~$\delta$ we considered so far
can be thought to be a constant-degree factor of a given polynomial~$f$ of degree~$d$.
Our goal would be to compute~$g$.
It is now easy to extend the above results  to hold for all
degree-$\delta$ factors of~$f$ simultaneously.

\begin{corollary}\label{cor:HIT-f}
Let $f(x,\z)$ be a polynomial of degree~$d$ with $n+1$ variables
that is monic in~$x$, and let~$\delta$ be a constant.
There exists  $\w,\w' \in \F^n$ with $w_i,w_i' \leq dn^{\poly(\delta)}$ such that
for any irreducible factor~$g$ of degree~$\delta$ of~$f$,
we have that~$\Psi(g)$ is an irreducible factor of~$\Psi(f)$.
\end{corollary}

\begin{proof}
The proof goes along the lines of Corollary~\ref{cor:HIT-g},
but we choose the weights slightly larger so that
the~$\widehat{G}(y)$ polynomials for \emph{all} the degree-$\delta$ factors~$g$ 
of~$f$ are non-zero simultaneously.
That is, we choose prime~$p$ in Lemma~\ref{lem:isolation} as $p = d n^{\poly(\delta)}$.
\end{proof}


\section{Computing the low-degree factors} \label{sec:low-degree}

We show how to compute the factors of constant degree of a given polynomial~$f$.
In Section~\ref{sec:constant-degree-promise} for the case when
\emph{all} factors of~$f$ have constant degree,
and in Section~\ref{sec:constant-degree} for general~$f$.
In both cases,
our algorithm starts by projecting the given polynomial
to a $3$-variate polynomial.
We start with this common part.


\subsection{Projected constant degree factors} 

The following algorithm is an initiating step in both, 
Algorithm~\ref{algo: promise-low-deg factors} and~\ref{algo: low-deg factors}
in Sections~\ref{sec:constant-degree-promise} and~\ref{sec:constant-degree}.
It takes an $n$-variate  polynomial~$f(\z)$ of degree~$d$.
The final goal is to compute the factors of~$f$ of degree~$\delta$,
for some given constant~$\delta$.
The initial steps are to project~$f$ to a trivariate monic polynomial 
and then factorize the projection.

In more detail,
the first step is to make~$f(\z)$ monic in a new variable~$x$ via Lemma~\ref{lem:monic}.
That is,
we compute~$\bal \in (\F \setminus \{0\})^n $ such that
the transformed polynomial $f_{\bal}(x,\z) = f(\bal x + \z)$ is monic. 

Then we apply Corollary~\ref{cor:HIT-f} to~$f_{\bal}(x,\z)$.
That is, 
we compute the weights $\w,\w' \in \F^n$ bounded by~$dn^{\poly(\delta)}$ 
and explicitly compute~$\Psi(f_{\bal}) \in \F[x,y,t]$ 
of degree at most $\widetilde{d} = d^2\, n^{\poly(\delta)}$.
Note that the $x$-degree of $f_{\bal}$ has not changed by mapping~$\Psi$.

The next step is to factor $3$-variate~$\Psi(f_{\bal})$
via Lemma~\ref{lemma:trivariate-factoring}.
Algorithm~1, {\sc Projected-Factoring}, below summarizes the steps.


\begin{algorithm}[htb]
  \caption{{\sc Projected-Factoring}}
  \label{algo: projected-low-deg factors}
  \SetKwInOut{Input}{Input}\SetKwInOut{Output}{Output}

  \Input{$f(\z) \in \calC_{\mu}(s,n,d)$ and a constant $\delta$.}
  \Output{All projected $3$-variate factors of $f$ of degree $\delta$ with their multiplicities.}
  \BlankLine

    Find $\bal$ such that $f_{\bal}(x,\z) = f(\bal x + \z~)$ is monic
    ~~~~~~{\tcc{by Lemma~\ref{lem:monic}}}

   Find $\w,\w' \in \F^n$ 
   as in Corollary~\ref{cor:HIT-f} and compute $\Psi(f_{\bal}) \in \F[x,y]$ in dense representation
      
   Factorize $\Psi(f_{\bal}) = h_1^{e_1} h_2^{e_2} \cdots h_m^{e_m}$ 
   ~~~~~~{\tcc{by Lemma~\ref{lemma:trivariate-factoring}, in dense representation}}

   Define $S_{\prf}= \set{h_i}{\deg_{x}(h_i) \le \delta}$ 
   and $S_{\prfm}=\set{(h_i,e_i)}{h_i \in S_{\prf}}$
   
\Return{$S_{\prf},S_{\prfm}$}
\end{algorithm}


For the running time of {\sc Projected-Factoring}
we have $nd \, \T_{\PIT(\HomC[\calC])}$ for finding~$\bal$ in step~1
by Lemma~\ref{lem:monic}.
Steps~2 and~3 take time $\poly(d\, n^{\poly(\delta)})$.
In summary,
the time complexity of {\sc Projected-Factoring} is
\begin{equation}\label{eq:time-ProjFac}
    nd \, \T_{\PIT(\HomC[\calC])} + \poly(d\, n^{\poly(\delta)}) \,.
\end{equation}


\subsection{Factors of a constant degree product} \label{sec:constant-degree-promise}

Given a polynomial that is the product of constant degree polynomials.
We show that all the factors can be computed in polynomial time.

\begin{theorem}\label{thm:constant-degree-promise}
For a constant~$\delta$,
let $\calD_{\delta} = \calP(n,\delta)$ and
$\calC_d \subseteq \calP(n,d)$ be the class of polynomials
that are a product of polynomials from~$\calD_{\delta}$,
i.e.\ $\calC_d = \prod \calD_{\delta}$.
Then $\Fac{\calC_d}{\calD_{\delta}}$ can be solved in time
$\T_{\Fac{\calC_d}{\calD_{\delta}}} =  \poly(dn^{\poly(\delta)})$.
\end{theorem}

\begin{proof}
We start by invoking {\sc Projected-Factoring} for~$f$ to get~$S'$,
the set of factors of the transformed trivariate 
polynomial~$\Psi(\tau_{\bal}(f))$  with their multiplicities.
Then we first invert the transformation~$\Psi$ on the factors
using Corollary~\ref{cor:HIT-g} and its remark.
Then we invert~$\tau_{\bal}$.
Since~$f$ and~$\Psi(\tau_{\bal}(f))$ have the same factorization pattern,
we finally get the factors of~$f$.
We summarize the steps in Algorithm~\ref{algo: promise-low-deg factors}.


\begin{algorithm}[htb]
  \caption{{\sc Constant-Degree-Factorization}}
  \label{algo: promise-low-deg factors}
  \SetKwInOut{Input}{Input}\SetKwInOut{Output}{Output}

  \Input{$f(\z) \in \calC_{\mu}(s,n,d)$ and a constant $\delta$
  such that \\ \emph{all} irreducible factors of $f$ have degree $\le \delta$.}
  \Output{All irreducible factors of $f$ with their multiplicities.}
  \BlankLine

  $L = \emptyset$~~~~~~  {\tcc{Initialize output list}}

   $(S_{\prf},S_{\prfm}) = $ {\sc Projected-Factoring}$(f,\delta)$ 
   
    \For{$(\widetilde{g},e) \in S_{\prfm}$}{
    \label{alg:promise:outer-for-loop}
     {\tcc{Computing the irreducible factors via inversion}}

    Compute $\widehat{g} = \Psi^{-1}(\widetilde{g})$, by Corollary~\ref{cor:HIT-g}
     
    Compute $g =\tau_{\bal}^{-1}(\widehat{g})$, and add $(g,e)$ to $L$
  }
  \Return{$L$}
\end{algorithm}


For the time complexity of {\sc Constant-Degree-Factorization},
recall that constant degree polynomials are sparse,
with sparsity $s = O(n^{\delta})$ by~(\ref{eq:M_delta-bound}).
Hence,
we get $\T_{\PIT(\HomC[\calC])} = \poly(dn^{\delta})$ by Theorem~\ref{thm:sparse-interpolation}.
The {\bf for}-loop with the inverse mappings
takes time~$\poly(dn^{\poly(\delta)})$.
\end{proof}

Note that Theorem~\ref{thm:constant-degree-promise} solves a \emph{promise case}.
That is,
we assume that~$f$ is a product of constant degree polynomials,
but we do not verify this assumption.
So in case that the assumption does not hold,
one can anyway run the algorithm from  Theorem~\ref{thm:constant-degree-promise},
but then it might output polynomials of degree~$\delta$ that are not factors of~$f$.
In Section~\ref{sec:constant-degree},
we show how to compute the constant degree factors of an arbitrary polynomial.


\subsection{Constant degree factors} \label{sec:constant-degree}

In Section~\ref{sec:constant-degree-promise},
we computed the constant degree factors of a polynomial~$f$
under the assumption that \emph{all} factors of~$f$ are of constant degree.
Now we skip the assumption and let~$f$ be a polynomial from some class
$\calC_{\mu} = \calC_{\mu}(s,n,d) \subseteq \calP(n,d)$.
We still want to compute the constant degree factors of~$f$ from
$\calD_{\delta} = \calP(n,\delta)$, for some constant~$\delta$.
We show that the problem can be reduced in polynomial time
to a PIT for~$\HomC[\calC_{\mu}]$ and divisibility tests~$\calC_{\mu}$ by~$\calD_{\delta}$.

\begin{theorem}\label{thm:constant-degree}
Let $\calC_{\mu}  = \calC_{\mu}(s,n,d)$  and 
$\calD_{\delta} = \calP(n,\delta)$, for a constant~$\delta$.
Then
$\Fac{\calC_{\mu}}{\calD_{\delta}}$
can be solved in time 
\[
\T_{\Fac{\calC_{\mu}}{\calD_{\delta}}} = 
nd \, \T_{\PIT(\HomC[\calC_{\mu}])} + d^2 n^{\poly(\delta)} \, \T_{\Div{\partial\calC_{\mu}}{\calD_{\delta}}} + \poly(s,n^{\poly(\delta)},d).
\]
\end{theorem}

\begin{proof}
Let~$f(\z) \in \calC$. 
To compute the factors of degree~$\delta$ of~$f$,
we start again by invoking {\sc Projected-Factoring} for~$f$ to get~$S$,
the set of factors of the transformed trivariate 
polynomial~$\Psi(\tau_{\bal}(f))$.
As in the proof of Theorem~\ref{thm:constant-degree-promise},
we compute the inverses of the transformations,
$g = \tau_{\bal}^{-1}(\Psi^{-1}(\widetilde{g}))$.

However~$\widetilde{g}$ might also \emph{not} 
correspond to a degree-$\delta$ factor of~$f$.
In this case,
either the inverse transformation does not go through properly,
or the degree we get is larger than~$\delta$.
In these cases, we can immediately throw away~$\widetilde{g}$; 
see the remark after Corollary~\ref{cor:HIT-g}.
But it could also be that we actually obtain a polynomial~$g$
of degree~$\delta$, 
just that it is not a factor of~$f$. 
For that reason, we finally do a divisibility check 
whether~$g \divides f$.
That way we will compute all factors of~$f$ of degree~$\delta$.
The multiplicities of the factors we compute via Lemma~\ref{lem:factor-multiplicity-reduction}.
We summarize the steps in Algorithm~\ref{algo: low-deg factors}.


\begin{algorithm}[htb]
  \caption{{\sc Constant-Degree-Factors}}
  \label{algo: low-deg factors}
  \SetKwInOut{Input}{Input}\SetKwInOut{Output}{Output}

  \Input{$f(\z) \in \calC_{\mu}(s,n,d)$ and a constant $\delta$.}
  \Output{All irreducible factors of $f$ of degree $\leq \delta$ with their multiplicities.}
  \BlankLine

  $L = \emptyset$, $L' = \emptyset$ ~~~~~{\tcc{Initialize output and intermediate candidates lists}}
  
  $(S_{\prf},S_{\prfm}) = $  {\sc Projected-Factoring}$(f,\delta)$
   
    \For{$\widetilde{g} \in S_{\prf}$}{
    \label{alg:non-promise:outer-for-loop}
     {\tcc{Computing candidate factors via divisibility}}

    Compute $\widehat{g} = \Psi^{-1}(\widetilde{g})$ (if the inverse exists) 
    of degree $ \leq \delta$ by Corollary~\ref{cor:HIT-g}
     
    Compute $g =\tau_{\bal}^{-1}(\widehat{g})$

    If $g \divides f$ then add $g$ to $L'$
  }
  \For{$g \in L'$} 
  {\tcc{Computing multiplicities via Lemma~\ref{lem:factor-multiplicity-reduction}}
    
    Let $z$ be a variable that $g$ depends on

    Find the smallest $e \geq 1$ such that 
    $g \not\divides\, \frac{\partial^e f}{\partial z^e}$ 
    and add $(g,e)$ to list $L$
    
  }
  \Return{$L$}
\end{algorithm}


For the time complexity of the factoring algorithm,
we have $nd \, \T_{\PIT(\HomC[\calC])}$ for transforming~$f$ 
to monic~$f_{\bal}$ by Lemma~\ref{lem:monic}.
Time $\poly(d\, n^{\poly(\delta)})$ is used for map~$\Psi$ and the factoring
of~$\Psi(f_{\bal})$.
Similar time is taken to invert and get the candidate factors. 
Finally,
we have at most $d^2 n^{\poly(\delta)}$ candidate polynomials~$g$ for which we test
divisibility of $g \divides f$ in Line~$6$.
For the multiplicities,
we first find a variable~$z$ that~$g$ depends on in Line~$8$.
Then we have at most~$d$ divisibility tests whether
$g \divides\, \frac{\partial^e f}{\partial z^e}$ in Line~$9$.
\end{proof}

As a consequence, 
we get a quasi-polynomial-time algorithm to compute the constant-degree factors of a sparse polynomial.
With a different proof,
this was already shown by Kumar, Ramanathan, and Saptharishi~\cite{kumar2024deterministic}.

\begin{corollary}[Constant-degree factors of sparse \cite{kumar2024deterministic}]
\label{cor:sparse-constant}
Let $\calC_{\sparse}  = \calC_{\sparse}(s,n,d)$  and 
$\calD_{\delta} = \calP(n,\delta)$, for a constant~$\delta$.
Then
$\Fac{\calC_{\sparse}}{\calD_{\delta}}$
can be solved in time 
\[
\T_{ \Fac{\calC_{\sparse}}{\calD_{\delta}}} = (snd)^{\poly(\delta) \log s}.
\]
\end{corollary}

\begin{proof}
Homogeneous components as well as the derivatives of a sparse polynomial remain sparse. 
Hence, we have
$\T_{\PIT(\HomC[\calC_{\sparse}])} = \poly(snd)$ by Theorem~\ref{thm:sparse-interpolation},
and
$\T_{\Div{\partial\calC_{\sparse}}{\calD_{\delta}}} = (snd)^{O(\delta \log s)}$
by Corollary~\ref{cor:sparse-by-constant-degree}.
Hence, 
we get the desired complexity by Theorem~\ref{thm:constant-degree}.
\end{proof}


Kumar, Ramanathan, and Saptharishi~\cite{kumar2024deterministic}
generalized sparse polynomials to polynomials computed by constant-depth circuits
and showed that the constant-degree factors can be computed in subexponential time.
We can now also derive this result via Theorem~\ref{thm:constant-degree}.

Recall that 
\[
\calC_{\cd{t}} (s,n,d) = \set{p \in \calP(n,d)}{p \text{ has a circuit of size $s$ and depth } t}.
\]
We state some properties of $\calC_{\cd{t}}$.
\begin{enumerate}\romanitems
    \item 
    $\calC_{\sparse}(s,n,d) \subseteq \calC_{\cd{2}}(s+1,n,d)$.
    \item
     $\calC_{\cd{t}}$ is closed for homogeneous components:
     For a polynomial $f \in \calC_{\cd{t}} (s,n,d)$,
     all homogeneous components of~$f$ are in $\calC_{\cd{(t+1)}} (sd,n,d)$
     (see \cite[Lemma 2.3]{oliveira2016factors}).
    \item 
    $\calC_{\cd{t}}$ is closed under derivatives:
For a polynomial $f \in \calC_{\cd{t}} (s,n,d)$,
a variable~$z$ and~$e \geq 1$,
we have
$\frac{\partial^e f}{\partial z^e} \in \calC_{\cd{(t+1)}} (d^2\,s,n,d)$ 
(see \cite[Lemma 2.5]{oliveira2016factors}).
\end{enumerate}

We apply Theorem~\ref{thm:constant-degree} to compute
constant-degree factors of constant-depth circuits. 

\begin{corollary}[Constant-degree factors of constant-depth \cite{kumar2024deterministic}]\label{cor:constant-depth-factoring}
Let $\calC_{\cd{t}} = \calC_{\cd{t}}(s,n,d)$ and
$\calD_{\delta} = \calP(n,\delta)$, for a constant~$\delta$.
Then, for any $\varepsilon >0$,
$\Fac{\calC_{\cd{t}}}{\calD_{\delta}}$ can be solved in time 
\[
\T_{ \Fac{\calC_{\cd{t}}}{\calD_{\delta}}} =  
\left(n\, (sd)^{O(t)}  \right)^{O((sd)^{\epsilon})} \,.
\]
\end{corollary}

\begin{proof}
Let $f \in \calC_{\cd{t}}$.
We consider the running times from Theorem~\ref{thm:constant-degree} to factor~$f$.
We argue that
\[
\T_{\PIT(\HomC[\calC_{\cd{t}}])}  
= \T_{\Div{\partial\calC_{\cd{t}}}{\calD_{\delta}}} 
= \left(n\, (sd)^{O(t)} \right)^{O((sd)^{\epsilon})}\,.
\]
For the homogeneous components of~$f$,
this follows from property~(ii) and Theorem~\ref{thm:PIT-constant-depth}.

For the divisibility test,
this follows from property~(i) and~(iii) above,
and Corollary~\ref{cor:constant-depth-by-sparse}.
\end{proof}


\subsection{Linear factors of $\ROABP$s} \label{sec:constant-degree-ROABP}

We can apply Theorem~\ref{thm:constant-degree} in the case of  
\emph{read-once oblivious arithmetic branching programs}, $\ROABP$s,
and compute \emph{linear} factors.
$\ROABP$s are arithmetic branching programs (ABPs) 
where there is an order on the variables
and on every path of the ABP,
every variable is evaluated in this order and only once.
For~$n$ variables $\z = (z_1,z_2, \dots, z_n)$, 
we describe the order by a permutation $\sigma \in S_n$
on the indices, 
$(z_{\sigma(1)}, z_{\sigma(2)}, \dots,z_{\sigma(n)}) $.
As size measure for $\ROABP$s, usually the width~$w$ is taken.
Its size as a graph is then bounded by~$nw$,
where~$n$ is the number of variables.
Width and size can greatly vary depending on the variable order.
Let
\[
\calC_{\ROABP_{\sigma}}(w,n,d) = \set{p \in \calP(n,d)}{p \text{ has an $\ROABP$ with order $\sigma$ of width } w}.
\]
When we skip the variable order~$\sigma$,
it means that the width can be achieved by \emph{some} order,
\[
\calC_{\ROABP(w,n,d)} = \bigcup_{\sigma \in S_n} \calC_{\ROABP_{\sigma}}(w,n,d)\,.
\]

$\ROABP$s generalize sparse polynomials:
$\calC_{\sparse}(s,n,d) \subseteq  \calC_{\ROABP_{\sigma}}(s,n,d)$, for any $\sigma \in S_n$.

For our arguments,
we need some folklore properties of $\ROABP$s. 

\begin{lemma}[Properties of $\ROABP$]\label{lem:closure-roabp}
Let $f \in \calC_{\ROABP_{\sigma}}(w_1,n,d)$ and $g \in \calC_{\ROABP_{\sigma}}(w_2,n,d)$.
Then
\begin{enumerate}\romanitems
    \item $f + g \in \calC_{\ROABP_{\sigma}}(w_1+w_2,n,d)$,
    \item $fg \in \calC_{\ROABP_{\sigma}}(w_1w_2,n,2d)$,
    \item $\Hom{d}{f} \in  \calC_{\ROABP_{\sigma}}((d+1)\,w_1,n,d)$,
    \item $\partial_{z^e}(f) \in \calC_{\ROABP_{\sigma}}(w_1,n,d-e)$,
    for a variable $z$ and $e \leq d$,
    \item 
    $(a_0 + a_1 z_1 + a_2 z_2 + \cdots + a_n z_n)^d \in  \calC_{\ROABP_{\sigma}}(nd+d+1,n,d)$,
    for $a_0,a_1, a_2, \dots, a_n \in \F$ 
    and any $\sigma \in S_n$ {\rm \cite{saxena2008diagonal}}.
\end{enumerate}
\end{lemma}

There are efficient PITs for $\ROABP$s in the literature.

\begin{lemma}\label{lem:PIT-ROABP}
Let $\calC_{\ROABP} = \calC_{\ROABP}(w,n,d)$.
Then $\PIT_{\calC_{\ROABP}}$ can be solved in 
\begin{itemize}
    \item 
polynomial time $\poly(ndw)$, in the whitebox setting~{\rm \cite{raz2005deterministic}}, 
\item 
quasi-polynomial time~$(ndw)^{O(\log \log w)}$,
in the blackbox setting~{\rm \cite{DBLP:journals/toc/GurjarKS17,guo2020improved}}. 
\end{itemize}
\end{lemma}

For factorization of a polynomial~$f$ computed by an $\ROABP$,
we apply Theorem~\ref{thm:constant-degree}.
The time for PIT of homogeneous components of~$f$
follows from Lemma~\ref{lem:closure-roabp}~(iii) and Lemma~\ref{lem:PIT-ROABP}.

For divisions,
we consider \emph{linear} polynomials.
This is the case $\delta = 1$ in the above terminology.
That is, we consider class~$\calD_1$.
We use Lemma~\ref{lem:divisibility} to reduce the divisibility test
to a PIT instance of a polynomial-size $\ROABP$,
which can be solved by Lemma~\ref{lem:PIT-ROABP}.
By Lemma~\ref{lem:closure-roabp}~(iv),
this holds similarly for the partial derivatives.

\begin{lemma}[$\ROABP$ by linear division]\label{lem:ROABP-divsion}
Let $\calC_{\ROABP} = \calC_{\ROABP}(w,n,d)$  and
$\calD_1 = \calP(n,1)$.
Then $\Div{\calC_{\ROABP}}{\calD_1}$ can be solved in time 
\begin{itemize}
    \item $\poly(ndw)$, in the whitebox setting,
    \item $(ndw)^{O(\log \log dw)}$, in the blackbox setting.
\end{itemize}
\end{lemma}

\begin{proof}
Let $f(\z) \in \calC_{\ROABP_{\sigma}}(w,n,d)$ and~$\ell(\z)$ be linear.
We apply Lemma~\ref{lem:divisibility}.
Let $S= [2d^2+1]$ and find an $\bal \in \F^n$ such that  $\ell(\bal) \ne 0$.
Such an~$\bal$ can be found in time~$O(n)$.
Let 
$c_{\beta,i} \in \F$ be constants such that
$\ell \divides f \iff \widehat{f} = f(\z+\bal) - \ell(\z+\bal) \, \tilde{h} = 0$, 
where
\begin{equation}\label{eq:h-tilde-ROABP}
\widetilde{h} =\sum_{\beta \in S} f(\beta \z + \bal) 
\sum_{0 \le i \le d}\,c_{\beta,i}\ell(\beta \z+\bal)^i \,.
\end{equation}
We consider the width of an $\ROABP$ that computes~$\widehat{f}$.
The terms~$\ell(\beta \z+\bal)^i$ in~(\ref{eq:h-tilde-ROABP}) 
have $\ROABP$s of width~$ni+i+1 \leq nd+d+1$ by Lemma~\ref{lem:closure-roabp}~(v) in any variable order,
and hence, in particular in the order~$\sigma$ of~$f$.
Applying Lemma~\ref{lem:closure-roabp}~(i) and~(ii),
we get that~$\widetilde{h}$ has an $\ROABP$ in order~$\sigma$ of width~$O(w n d^4)$.
Hence, for~$\widehat{f}$,
we get an $\ROABP$ in order~$\sigma$ of width~$O(w^2 n^2 d^4)$.
Now the claim follows from Lemma~\ref{lem:PIT-ROABP}.
\end{proof}

We plug the time bounds from Lemma~\ref{lem:PIT-ROABP} and~\ref{lem:ROABP-divsion}
in Theorem~\ref{thm:constant-degree}.

\begin{corollary}[Linear factors of $\ROABP$s]\label{cor:roabp-linfactor}
Let $\calC_{\ROABP} = \calC_{\ROABP}(w,n,d)$  and
$\calD_1 = \calP(n,1)$.
Then
$\Fac{\calC_{\ROABP}}{\calD_1}$ can be solved in time

\begin{itemize}
    \item 
    $\poly(ndw)$ in the whitebox setting,
    \item 
    $(n dw)^{O(\log \log dw)}$ in the blackbox setting. 
\end{itemize}
\end{corollary}

When the input $f$ is $s$-sparse, this result is already known due to Volkovich~\cite[Theorem~4]{volkovich2015deterministically}.


\section{Computing the sparse factors} 
\label{sec:sparse-factors}

Recall the class of sparse polynomials,
\begin{equation}
\calC_{\sparse}(s,n,d) = \set{p \in \calP(n,d)}{ \sparse(p) \leq s }.
\end{equation}
We want to compute the sparse factors of a given polynomial.
For a polynomial~$f$ from a class~$\calC_{\mu}$,
we show that the sparse factors of~$f$ can be computed efficiently
relative to
a PIT for~$\HomC[\calC_{\mu}]$,
a divisibility test of~$f$ by sparse candidate factors,
and an irreducibility preserving projection of sparse polynomials,
$\Irr(\calC_{\sparse})$ (see Section~\ref{sec:Hilbert}).

\begin{lemma}\label{lem:irred-test}
Given $f \in \calC_{\sparse}$,
we can test whether~$f$  is irreducible
in time $ \poly(snd) + \T_{\Irr(\calC_{\sparse})}$.
\end{lemma}

\begin{proof}
We first transform~$f(\z)$ via Lemma~\ref{lem:monic} into a polynomial~$f_\bal(x,\z)$ that is monic in the new variable~$x$.
Note that~$\bal \in (\F\setminus\{0\})^n$ can be computed
in time~$\poly(snd)$ by Theorem~\ref{thm:sparse-interpolation}.

Then,
by solving $\Irr(\calC_{\sparse})$, one can find a set~$\calH_{\bal}$ such 
that~$f(\z)$ is reducible if and only 
if~$f(\bal x + \bbet t + \bgam)$ is reducible, 
for every $(\bbet,\bgam) \in \calH_{\bal}$. 
Whether a bivariate polynomial is reducible can be checked in time~$\poly(d)$ by Lemma~\ref{lemma:trivariate-factoring}.  
\end{proof}

\begin{theorem}[Sparse factors] \label{thm:sparse-factors}
Let $\calC_{\mu} = \calC_{\mu}(s,n,d)$ 
and  $\calC_{\sparse} = \calC_{\sparse}(s,n,d)$.
Then 
$\Fac{\calC_{\mu}}{\calC_{\sparse}}$ can be solved in time 
\[
\T_{\Fac{\calC_{\mu}}{\calC_{\sparse}}} = 
\poly(snd)\, (\T_{\PIT(\HomC[\calC_{\mu}])} \,+\,   \T_{\Irr(\calC_{\sparse})} \, + \, \T_{\Div{\partial\calC_{\mu}}{\calC_{\sparse}}}) \,.
\]
\end{theorem}

\begin{proof}
Let $f(\z) \in \calC_{\mu}(s,n,d)$ and let~$g$ be a $s$-sparse irreducible factor of~$f$ with multiplicity~$e$, 
that is~$f=g^{e} h$, 
where $\gcd(g,h)=1$. 
Let $\deg(g)=d_g$ and $\deg(h)=d_h$.

The first step is to transform~$f$ to a monic polynomial
by finding an $\bal \in (\F\setminus\{0\})^n$ such that $\Hom{d}{f}(\bal) \ne 0$. 
Observe that $\Hom{d}{f}= (\Hom{d_g}{g})^e \, \Hom{d_h}{h}$. 
Therefore, we also have $\Hom{d_g}{g}(\bal) \ne 0$.

The second step is to solve $\Irr(\calC_{\sparse})$ and compute
a set~$\calH_{\bal}$ such that $g(\bal x + \bbet t + \bgam)$ remains irreducible for some~$(\bbet,\bgam) \in \calH_{\bal}$. 
We call such a $(\bbet,\bgam)$ a \emph{good point} for~$g$.

Since we do not know~$g$ for now,
we also do not know which pair $(\bbet,\bgam) \in \calH_{\bal}$ is a good point for~$g$.
For that reason,
we iteratively try all pairs $(\bbet,\bgam) \in \calH_{\bal}$ in the following.
Our algorithm will detect when a pair is not good at some point and
then ignore that pair.
Consider a good point $(\bbet,\bgam) \in \calH_{\bal}$
and compute bivariate polynomial~$\widehat{f}$,
\[
\widehat{f}(x,t) = f(\bal x + \bbet t + \bgam)
\]
in dense representation by interpolation.
Then factor~$g$ is mapped to  $\widehat{g}(x,t) = g(\bal x + \bbet t + \bgam)$,
a factor of~$\widehat{f}$.
We factorize~$\widehat{f}$ over~$\F$ by Lemma~\ref{lemma:trivariate-factoring}.
Let $F = \{\widehat{g}_1,\widehat{g}_2,  \ldots, \widehat{g}_m\}$ be the set of irreducible factors of~$\widehat{f}$.
Note that for a good point $(\bbet,\bgam)$,
we have $\widehat{g} \in F$.

Our task now is to find out which polynomials in~$F$ are coming from sparse factors of~$f$
and to recover these factors.
Let $\calE = \calE(s,n,d) \subseteq \F^n$ be the set of evaluation points to interpolate
$s$-sparse $n$-variate polynomials of degree~$d$ from Theorem~\ref{thm:sparse-interpolation}.
For any $\cb \in \calE$,
we compute the trivariate polynomial~$\widehat{f}_{\cb}$,
\[
\widehat{f}_{\cb}(x,t_1,t_2)  = f(\bal x + \bbet t_1 + (\cb - \bgam) t_2 + \bgam)
\]
in dense representation by interpolation.
Factor~$g$ is mapped correspondingly to
$\widehat{g}_{\cb}(x,t_1,t_2) = g(\bal x + \bbet t_1 + (\cb - \bgam) t_2 + \bgam)$.
Note that $\widehat{f}_{\cb}$ and $\widehat{g}_{\cb}$
are monic polynomial in~$x$ and
$\widehat{g}_{\cb}$ remains irreducible, 
since~$\widehat{g}_{\cb}(x,t,0)= g(\bal x + \bbet t + \bgam)$ is irreducible,
for a good point $(\bbet,\bgam)$.
 
We factorize~$\widehat{f}_{\cb}$ over~$\F$ by Lemma~\ref{lemma:trivariate-factoring}.
Let $F_{\cb} = \{\widehat{g}_{\cb,1},\widehat{g}_{\cb,2},  \ldots, \widehat{g}_{\cb,r}\}$ 
be the set of irreducible factors of~$\widehat{f}_{\cb}$.
Then we have $\widehat{g}_{\cb} \in F_{\cb}$.

Note that $\widehat{f}$ and $\widehat{g}$ are the projections of $\widehat{f}_{\cb}$ and $\widehat{g}_{\cb}$
for $t_2 = 0$. That is
\begin{align}
   \widehat{f}(x,t) &=  \widehat{f}_{\cb}(x,t,0),\\
    \widehat{g}(x,t) &= \widehat{g}_{\cb}(x,t,0).
\end{align}
The reason that we introduced the trivariate polynomials is
that we can use them to get evaluation points of~$g$:
We have
\begin{equation}
    \widehat{g}_{\cb}(0,0,1) = g(\cb).
\end{equation}
This we will use to get~$g$ via sparse interpolation by Theorem~\ref{thm:sparse-interpolation}.

Recall that we know that $\widehat{g} \in F$ and $\widehat{g}_{\cb} \in F_{\cb}$.
However,
we first need to know which polynomial in~$F_{\cb}$ corresponds to~$\widehat{g}$,
for every $\cb \in \calE$.
The polynomials in~$F$ are used as reference.
That is,
for every $\widehat{g}_j \in F$, 
we compute a list~$L_j$ of pairs
\begin{equation}
    L_j = \set{(\cb,\Delta)}{\cb \in \calE \text{ and } \exists \widehat{g}_{\cb,i} \in F_{\cb}~~
    \widehat{g}_j(x,t) = \widehat{g}_{\cb,i}(x,t,0) \text{ and } \Delta = \widehat{g}_{\cb,i}(0,0,1) }
\end{equation}
Observe that when  $\widehat{g} =  \widehat{g}_j$,
then~$L_j = \set{(\cb,g(\cb))}{\cb \in \calE}$ 
is a list of point-values for~$g$.
Hence we can compute~$g$ by interpolation from the points in~$L_j$ by Theorem~\ref{thm:sparse-interpolation}.

In case we consider a polynomial~$\widehat{g}_{j'} \in F$ such that
$\widehat{g} \not=  \widehat{g}_{j'}$ and do interpolation for list~$L_{j'}$,
we might still end up in a sparse polynomial, say~$g'$,
but $g'$ might be reducible or not be a factor of~$f$.
Hence,
after interpolation,
we check whether~$g'$ is reducible via Lemma~\ref{lem:irred-test}
and whether~$g'$ divides~$f$ to rule out the wrong choices.

Once we have the correct factor~$g$ in hand,
we compute its multiplicity via Lemma~\ref{lem:factor-multiplicity-reduction}.
Algorithm~\ref{algo:sparse factors} summarizes the steps.


\begin{algorithm}[htb]
  \caption{{\sc Sparse-Factors}}
  \label{algo:sparse factors}
  \SetKwInOut{Input}{Input}\SetKwInOut{Output}{Output}

  \Input{$f(\z) \in \calC_{\mu}(s,n,d)$.}
  \Output{All irreducible $s$-sparse factors of $f$ with their multiplicities.}
  \BlankLine

    $L = \emptyset$, $L' = \emptyset$ 
    ~~~~~{\tcc{Initialize output list and intermediate candidates list}}

  Find $\bal \in (\F\setminus\{0\})^n$ such that $f_{\bal}(x,\z)$ is monic in~$x$
  ~~~~~{\tcc{by Lemma~\ref{lem:monic}}}
  
  Compute $\calH_{\bal} = \Irr(\calC_{\sparse})(\bal)$
  
  \For{each~$(\bbet,\bgam) \in \calH_{\bal}$}{
   Compute $\widehat{f}(x,t) = f(\bal x + \bbet t + \bgam)$ in dense representation
   
   Factorize $\widehat{f}$.
   Let $F = \{\widehat{g}_1, \widehat{g}_2,\ldots, \widehat{g}_m \}$ be the set of its irreducible factors

   $L_j = \emptyset$, for $j \in [m]$     ~~~~~{\tcc{Initialize point-value lists}}
 
   Let $\calE$ be the evaluation points for $\calC_{\sparse}(s,n,d)$ 
   ~~~~~{\tcc{see Theorem~\ref{thm:sparse-interpolation}}}
   
  \For{each $\cb \in \calE$}{
  
   Compute $\widehat{f}_{\cb}(x,t_1,t_2)  = f(\bal x + \bbet t_1 + (\cb - \bgam) t_2 + \bgam)$  in dense representation

   Factorize $\widehat{f}_{\cb}$. 
   Let $F_{\cb} = \{\widehat{g}_{\cb,1},\widehat{g}_{\cb,2},  \ldots, \widehat{g}_{\cb,r}\}$  be the set of its irreducible factors
  
     \For{each $j \in [m]$}{ 
      
        Search for $i \in [r]$ such that $\widehat{g}_j (x,t) = \widehat{g}_{\cb,i}(x,t,0)$

        $\Omega = \widehat{g}_{\cb,i}(0,0,1)$

        add $(\cb,\Omega)$ to $L_j$
      }
    }
   
  \For{each $j \in [m]$}{
{\tcc{Compute irreducible factors using Theorem~\ref{thm:sparse-interpolation}
and Lemma~\ref{lem:irred-test}}}
Compute polynomial $P_j$ by sparse interpolation on the points in $L_j$

If $P_j$ is $s$-sparse, irreducible  and $P_j \divides f$ then add $P_j$ to $L'$ 
  }
}

\For{each $P\in L'$}{
{\tcc{Compute multiplicities via Lemma~\ref{lem:factor-multiplicity-reduction}}}

    Let $z$ be a variable that $P$ depends on

    Find the smallest $e \geq 1$ such that $P \not\divides \, \frac{\partial^e f}{\partial z^e}$ and add $(P,e)$ to list $L$
  }
  \Return{L}
  
\end{algorithm}


The time complexity is now easy to see.
Steps~2 and~3 take time
$\T_{\PIT(\HomC[\calC_{\mu}])}$ and $\T_{\Irr(\calC_{\sparse})}$, respectively.
The {\bf for}-loop line 4~-~18 is executed $|\calH_{\bal} | \leq \T_{\Irr(\calC_{\sparse})}$
times.
Steps~5~-~17 take time $\poly(snd)$ by Lemma~\ref{lemma:trivariate-factoring}
and Theorem~\ref{thm:sparse-interpolation}.
Line~18 takes time $\poly(snd) + \T_{\Irr(\calC_{\sparse})}$ to test irreducibility by Lemma~\ref{lem:irred-test},
plus the time to check divisibility, $\T_{\Div{\calC_{\mu}}{\calC_{\sparse}}}$.

In Line~20, since we have~$P$ computed explicitely,
we can choose a variable~$z$ that occurs in~$P$.
Then we have at most~$d$  divisibility tests with the derivatives of~$P$ in Line~21.
In summary,
we get the time bound claimed in the theorem statement.

\end{proof}

Let $\calD \subseteq C_{\sparse}$ be a class of polynomials that are sparse.
We observe that the proof of Theorem~\ref{thm:sparse-factors} works completely analogous
when we replace~$C_{\sparse}$ by~$\calD$.
The only change is in line~18 of {\sc Sparse-Factors}
when we have computed a sparse representation of~$P_j$.
There we also need to check the membership of~$P_j$ in~$\calD$.
This is mostly trivial, for example for $\calD = \calD_{\delta}$.
However,
in the general setting $\calD \subseteq \calC_{\sparse}(s,n,d)$,
we need to assume that, given a polynomial in sparse representation,
membership in~$\calD$ can be efficiently decided,
i.e.~in time~$\poly(snd)$.

\begin{corollary}\label{cor:sparse-factors}
    Let $\calC_{\mu} = \calC_{\mu}(s,n,d)$ 
    and  $\calD \subseteq \calC_{\sparse}(s,n,d)$ with efficient membership tests.
Then 
$\Fac{\calC_{\mu}}{\calD}$ can be solved in time 
\[
\T_{\Fac{\calC_{\mu}}{\calD}} = 
\poly(snd) \, (\T_{\PIT(\HomC[\calC_{\mu}])}  \,+\,  \T_{\Irr(\calD)} \,+\, \T_{\Div{\partial \calC_{\mu}}{\calD}} )   \,.
\]
\end{corollary}

By choosing $\calD = \calD_\delta$, the polynomials of degree~$\delta$
in Corollary~\ref{cor:sparse-factors},
we get an alternative way to derive Theorem~\ref{thm:constant-degree}.
Just observe that $\T_{\Irr(\calD_\delta)} = \poly(dn^{\poly(\delta)})$.
In the following,
we show further applications of Theorem~\ref{thm:sparse-factors} and Corollary~\ref{cor:sparse-factors}.


\subsection{Sparse factors of constant-depth circuits}

We apply Theorem~\ref{thm:sparse-factors} to output sparse factors of constant-depth circuits. 
\begin{corollary}\label{cor:sparse-factoring}
Let $\calC_{\cd{t}} = \calC_{\cd{t}}(s,n,d)$ and
$\calC_{\sparse} = \calC_{\sparse}(s,n,d)$.
Then, for any $\varepsilon >0$,
$\Fac{\calC_{\cd{t}}}{\calC_{\sparse}}$ can be solved in time 
\[
\T_{\Fac{\calC_{\cd{t}}}{\calC_{\sparse}}} =
\left(n\, (sd)^{O(t)}  \right)^{O((sd)^{\epsilon})} \, \T_{\Irr(\calC_{\sparse})} \,.
\]
\end{corollary}

\begin{proof}
Let $f \in \calC_{\cd{t}}$.
We consider the running times  from Theorem~\ref{thm:sparse-factors} to factor~$f$.
We can use a similar argument as in the proof of Corollary~\ref{cor:constant-depth-factoring}.
For PIT, we have
$\T_{\PIT(\HomC[\calC_{\cd{t}}])} = \left(n\, (sd)^{O(t)} \right)^{O((sd)^{\epsilon})}$ 
and the same time bound we have for 
 $\T_{\Div{\partial\calC_{\cd{t}}}{\calC_{\sparse}}}$ by Corollary~\ref{cor:constant-depth-by-sparse}.
\end{proof}

After we submitted this paper,
Bhattacharjee et al.~\cite{bhattacharjee2025closure} showed 
a subexponential bound on $\T_{\Irr(\calC_{\sparse})}$.
It follows that
the sparse factors of constant-depth circuits can be computed 
in subexponential time.

\begin{corollary}[\cite{bhattacharjee2025closure}] \label{cor:sparse-factoring-subexp}
Let $\calC_{\cd{t}} = \calC_{\cd{t}}(s,n,d)$ and
$\calC_{\sparse} = \calC_{\sparse}(s,n,d)$.
Then, for any $\varepsilon >0$,
$\Fac{\calC_{\cd{t}}}{\calC_{\sparse}}$ can be solved in time 
\[
\T_{\Fac{\calC_{\cd{t}}}{\calC_{\sparse}}} =
\left(n\, (sd)^{O(t)}  \right)^{O((sd)^{\epsilon})} \,  \,.
\]
\end{corollary}

In fact,
Bhattacharjee et al.~\cite{bhattacharjee2025closure} showed 
that constant-depth circuits can be completely factored
in subexponential time.


\subsection{Sum-of-univariate factors of sparse polynomials}

Let us consider the family of polynomials that can be written as a sum of univariate polynomials,
\begin{equation}
\calC_{\SU}(n,d) =   \set{\sum_{i=1}^n\,p_i(z_i)}{p_i \in \calP(1,d), \text{ for } i = 1,2,\dots,n}\,.
\end{equation}
Note that $\calC_{\SU}(n,d) \subsetneq \calC_{\sparse}(nd+1,n,d)$. 

Given a sparse polyomial,
we show that its factors that are sums of univariates can be computed in polynomial time.
The result was already shown by Volokovich~\cite{volkovich2015deterministically}
with a different technique.
We show that it also follows via Corollary~\ref{cor:sparse-factors}.

We already know that PIT for sparse polynomials is in polynomial time.
Moreover,
Saha, Saptharishi, and Saxena~\cite{saha2013case} showed that
divisibility of a sparse polynomial by a sum of univariates is in polynomial time. 

\begin{theorem}[{\cite[Theorem~5.1]{saha2013case}}]\label{thm:Div-sparse-by-su}
Let  
$\calC_{\sparse} = \calC_{\sparse}(s,n,d)$ and 
$\calC_{\SU} = \calC_{\SU}(n,d)$.
Then   
$\Div{\calC_{\sparse}}{\calC_{\SU}}$ 
can be computed in time~$\T_{\Div{\calC_{\sparse}}{\calC_{\SU}}} = \poly(snd)$.
\end{theorem}
Therefore,
it suffices to give bounds on the time for irreducible projection~$\Irr(\calC_{\SU})$
(see Section~\ref{sec:Hilbert}).
We show that $\T_{\Irr(\calC_{\SU})}$ is polynomially bounded.
We crucially use the following theorem.

\begin{theorem}[{\cite[Theorem~5.2]{saha2013case}}]\label{thm:irreducible-su}
Let $p \in \calC_{\SU}(n,d)$ with $|\var(p)| \ge 3$.
Then~$p$ is irreducible.
\end{theorem}

\begin{lemma}\label{lem:SU-irred-proj}
$\T_{\Irr(\calC_{\SU})}\;\le\; O(n^3d^{30})$.
\end{lemma}

\begin{proof}
Let $p(\z) = \sum_{i \in [n]} p_i(z_i) \in \calC_{\SU}(n,d)$ be irreducible
and $\bal \in (\F\setminus\{0\})^n$ such that $p(\bal x + \z)$ is monic in~$x$.
In order to apply Theorem~\ref{thm:irreducible-su},
we consider~$\var(p)$, the variables that~$p$ depends on.

If $|\var(p)| \leq 2$, we can directly derandomize~Theorem~\ref{thm:HIT}:
Suppose $\var(p) = \{z_1,z_2\}$
so that $p(\z) = p(z_1,z_2)$.
There is a polynomial~$P$ in~$2 \cdot 2 = 4$ variables of degree~$2d^5$
such that any  non-root of~$P$ gives an irreducible projection of~$p$,
where we set the other variables of~$p$ to zero,
$z_i = 0$ for $i=3,4, \dots, n$.
Hence,
the trivial hitting set~(\ref{eq:hit}) of size $(2d^5 +1)^4 = O(d^{20})$
can be used for PIT for~$P$.
Since we want to handle the blackbox case,
we do not assume that we have~$p$ in hand.
Hence,
we do not know whether actually $\var(p) = \{z_1,z_2\}$.
So finally we will try all $\binom{n}{2}$ possibilities to choose~$2$ variables.
For each possibility, say $\{z_i,z_j\}$,  we take the same hitting set,
but put the values at~$z_i$ and~$z_j$ and set the remaining $z$-variables to~$0$.
That way we get~$\binom{n}{2}$ hitting sets,
and take their union to get the final hitting set of size~$O(n^2 d^{20})$.

Now let $|\var(p)| \geq 3$.
Similar as in the above case,
assume for now that $\{z_1,z_2,z_3\} \subseteq \var(p)$,
and in the end, 
we will try all $\binom{n}{3}$ possibilities to choose~$3$ variables.

Consider the substitution
$\psi: z_i \mapsto \alpha_i x$, for $i \geq 4$. 
Then
$\psi(p) \in \calC_{\SU}(4,d)$, 
since $\sum_{i=4}^n p_i(\alpha_i x)$ is a univariate polynomial in~$x$. 
Because $|\var(\psi(p)| \geq 3$,
polynomial~$\psi(p)$ is  irreducible by~Theorem~\ref{thm:irreducible-su}.

Now we substitute also the first~$3$ variables by
$\phi: z_i \mapsto \alpha_i x + z_i$, for $i \in [3]$. 
Then 
\[
\widehat{p}(x,z_1,z_2,z_3) =\phi(\psi(p)) = p_1(\alpha_1 x+ z_1) + p_2(\alpha_2 x+ z_2) + p_3(\alpha_3 x+ z_3) + \sum_{i=4}^n p_i(\alpha_i x)
\]
is also irreducible and monic in~$x$. 
Therefore, 
by Theorem~\ref{thm:HIT},  
there exists a polynomial~$P$ of degree $2d^5$ in $2 \cdot 3 = 6$ variables
such that any  non-root of~$P$ gives an irreducible projection of~$p$,
where we set the other variables of~$p$ to zero,
$z_i = 0$ for $i \geq 4$.
Hence,
we may again take the trivial hitting set~(\ref{eq:hit}) of size $(2d^5 + 1)^6 = O( d^{30})$.
Now we can argue similarly as above
and build the union of hitting sets over all~$\binom{n}{3}$ possibilities 
to choose~$3$ variables.
Hence,
we end up with a hitting set of size~$O(n^3 d^{30})$.
\end{proof}

Theorem~\ref{thm:sparse-interpolation},
Theorem~\ref{thm:Div-sparse-by-su}, and 
Lemma~\ref{lem:SU-irred-proj}
give all the time complexities used in Corollary~\ref{cor:sparse-factors}.
We conclude that the sum-of-univariate factors of a sparse polynomial can be computed in polynomial time.

\begin{corollary}[{\cite{volkovich2015deterministically} Sum-of-univariate factors}]\label{cor:su-factoring}
Let  
$\calC_{\sparse} = \calC_{\sparse}(s,n,d)$ and 
$\calC_{\SU} = \calC_{\SU}(n,d)$.
Then   
$\Fac{\calC_{\sparse}}{\calC_{\SU}}$ 
can be computed in time~$\T_{\Fac{\calC_{\sparse}}{\calC_{\SU}}} = \poly(snd)$.
\end{corollary}


\section*{Conclusion}

Finally, we collect some open questions.

\begin{enumerate}
\item Can we decide whether a given sparse polynomial is irreducible in deterministic quasi-polynomial time? The proof may already give a good bivariate projection that preserves irreducibility. Then Theorem~\ref{thm:sparse-factors} would give us a deterministic quasi-polynomial time algorithm to find the irreducible sparse factors of a sparse polynomial.

\item 
Can we find \emph{bounded individual degree} sparse factors of a sparse polynomial (that may not have bounded individual degree) in deterministic quasi-polynomial time? 
Volkovich asked if multilinear factors of a sparse polynomial can be found in deterministic polynomial 
time~\cite{volkovich2015deterministically}. 
Unfortunately, we do not have any deterministic polynomial time algorithm even to test if a sparse polynomial is divisible by a multilinear polynomial.

\item Given a blackbox computing the product of sparse irreducible polynomials $f_i$ with bounded individual degree, find the factors $f_i$ in deterministic polynomial/quasi-polynomial time. 
Bhargava, Saraf, and Volkovich~\cite{bhargava2020deterministic} gave a quasi-polynomial time algorithm, when the input polynomial is sparse with bounded individual degree.

\item Finally, we highlight a question from \cite{bhargava2020deterministic}.
Given a blackbox computing the product of two (or more) unknown sparse irreducible polynomials, find those irreducible factors in deterministic quasi-poly or polynomial time.

\end{enumerate}

\subsection*{Acknowledgements} 
P.D.\ is supported by the SUG Grant (\#025774-00001) 
\emph{Debordering and Derandomization in Algebraic Complexity}
funded by Nanyang Technological University. 
T.T's research supported by DFG grant TH 472/5-2. 
A.S.\ was supported by DFG grant TH 472/5-1 when the work started.

\bibliographystyle{alpha}
\bibliography{toc}

\end{document}